\documentclass[11pt]{article}
\usepackage[top=1in,left=1in,right=1in,bottom=1in]{geometry}
\pdfoutput=1

\newfont{\mycrnotice}{ptmr8t at 7pt}
\newfont{\myconfname}{ptmri8t at 7pt}
\let \originalleft \left
\let\originalright\right
\renewcommand{\left}{\mathopen{}\mathclose\bgroup\originalleft}
\renewcommand{\right}{\aftergroup\egroup\originalright}

\usepackage{color,array}
\usepackage{amsthm}
\usepackage{paralist}
\usepackage{booktabs}
\usepackage{paralist}

\usepackage[colorlinks=true, citecolor=black, urlcolor=cyan,
linkcolor=black]{hyperref}
\usepackage{amsfonts,amsmath,algorithm,algorithmicx,amssymb,graphicx,url,multirow}
\usepackage{enumerate}
\usepackage[noend]{algpseudocode}
\usepackage[bf,small]{caption}
\usepackage{enumitem}
\usepackage{pgfplots}
\usepackage{subcaption}
\usepackage{makecell}

\DeclareCaptionType{copyrightbox}

\newcommand{\myparagraph}[1]{\smallskip\noindent {\bf #1.}}
\newcommand{\defn}[1]{\emph{\textbf{#1}}}

\newcommand{\mpram}{MT-RAM}

\newcommand{\process}{thread}
\newcommand{\processes}{threads}

\newcommand{\forkins}{\texttt{fork}}
\newcommand{\insend}{\texttt{end}}

\def\code#1{\texttt{#1}}
\def\N{\mathbb{N}}
\providecommand{\ceil}[1]{\left\lceil#1\right\rceil}
\providecommand{\on}[1]{\operatorname{#1}}
\providecommand{\set}[1]{\left\{#1\right\}}

\newtheorem{theorem}{Theorem}
\newcommand{\codevar}[1]{\mathit{#1}}

\algdef{SE}[DOWHILE]{Do}{doWhile}{\algorithmicdo}[1]{\algorithmicwhile\ #1}%
\algdef{SE}[SUBALG]{Indent}{EndIndent}{}{\algorithmicend\ }%
\algtext*{Indent}
\algtext*{EndIndent}

\usepackage{microtype}
\begin{document}
\date{}

\title{Batch-Parallel Euler Tour Trees\footnote{This is the full version of the paper
    appearing in ALENEX 2019.}
}
\author{Thomas
Tseng\thanks{Computer Science Department, Carnegie Mellon University. \href{mailto:thomasts@alumni.cmu.edu}{thomasts@alumni.cmu.edu}} \\
\and
Laxman Dhulipala\thanks{Computer Science Department, Carnegie Mellon University.
\href{mailto:ldhulipa@cs.cmu.edu}{ldhulipa@cs.cmu.edu}} \\
\and
Guy Blelloch\thanks{Computer Science Department, Carnegie Mellon University.
\href{mailto:guyb@cs.cmu.edu}{guyb@cs.cmu.edu}}
}
\date{}

\maketitle

\begin{abstract} \small
  The dynamic trees problem is to maintain a forest undergoing edge
  insertions and deletions while supporting queries for information
  such as connectivity. There are many existing data structures for this
  problem, but few of them are capable of exploiting parallelism in the
  batch setting, in which large batches of edges are inserted or deleted from
  the forest at once. In this paper, we demonstrate that the Euler tour tree, an
  existing sequential dynamic trees data structure, can be parallelized in the
  batch setting.  For a batch of $k$ updates over a forest of $n$ vertices, our
  parallel Euler tour trees achieve $O(k \log (1 + n/k))$ expected work and
  $O(\log n)$ depth with high probability. Our work bound is asymptotically
  optimal, and we improve on the depth bound achieved by Acar et al.\ for
  the batch-parallel dynamic trees problem~\cite{acar2017brief}.

  Our main building block for parallelizing Euler tour trees is a
  batch-parallel skip list data structure, which we believe may be of
  independent interest.  Euler tour trees require a sequence data
  structure capable of joins and splits. Traditionally, balanced
  binary trees are used, but they are difficult to join or split
  in parallel when processing batches of updates. We show that skip
  lists, on the other hand, support batches of joins or splits of size
  $k$ over $n$ elements with $O(k \log (1 + n/k))$ work in expectation
  and $O(\log n)$ depth with high probability. We also achieve the
  same efficiency bounds for augmented skip lists, which allows us to
  augment our Euler tour trees to support subtree queries.

  Our data structures achieve between $67$--$96\times$ self-relative
  speedup on 72 cores with hyper-threading on large batch sizes. Our
  data structures also significantly outperform the fastest existing
  sequential dynamic trees data structures empirically.
\end{abstract}

\clearpage

\section{Introduction}\label{sec:intro}

In the dynamic trees problem proposed by Sleator and
Tarjan~\cite{sleator1983data}, the objective is to maintain a forest
that undergoes \emph{link} and \emph{cut} operations. A link
operation adds an edge to the forest, and a
cut operation deletes an edge. Additionally, we want to maintain useful
information about the forest. Most commonly we are concerned with whether pairs
of vertices are connected, but we might also be interested in properties like
the size of each tree in the forest. Sleator and Tarjan first studied the
dynamic trees problem in order to develop fast network flow
algorithms~\cite{sleator1983data}.  Dynamic trees are also an important
component of many dynamic graph algorithms~\cite{sleator1983data,
henzinger1995randomized, holm2001poly, alstrup2005maintaining,
kapron2013dynamic}.

In the batch-parallel version of the dynamic trees problem, the
objective is to maintain a forest that undergoes \emph{batches} of
link and cut operations. Though many sequential data
structures exist to maintain dynamic trees, to the best of our
knowledge, the only batch-parallel data structure is a very recent
result by Acar et al.~\cite{acar2017brief}. Their data structure is
based on parallelizing RC-trees, which require transforming the represented
forest to have bounded degree in order to perform
efficiently~\cite{acar2004dynamizing}. Obtaining a data structure
without this restriction is therefore of interest. Furthermore, it is
of intellectual interest whether the arguably simplest solution to the
dynamic trees problem, Euler tour trees (ETTs), can be
parallelized.

In this paper, we answer this question in the affirmative and show
that Euler tour trees, a data structure introduced
by Henzinger and King~\cite{henzinger1995randomized} and Miltersen et
al.~\cite{miltersen1994complexity}, achieve asymptotically optimal
work and optimal depth in the batch-parallel setting. We also develop a
batch-parallel skip list upon which we build our Euler tour trees. Note that
batching is not only useful for parallel applications but also for
single-threaded applications; our $O(k \log (1 + n/k))$ work bounds for $k$
operations over $n$ elements
on Euler tour trees and augmented skip lists beat the $O(k \log n)$
bounds achieved by performing each operation one at a time on standard
sequential data structures.

Our main contributions are as follows:

\noindent\textbf{Skip lists for simple, efficient parallel joins and parallel
splits.} We show that we can perform $k$ joins or
$k$ splits over $n$ skip list elements with $O(k \log (1 + n/k))$ expected
work and $O(\log n)$ depth with high probability\footnote{
We say that an algorithm has $O(f(n))$ cost \defn{with high
probability (w.h.p.)} if it has $O(\on{poly}(k)\cdot f(n))$ cost with probability
at least $1 - 1/n^{k}$.}. To the best of our knowledge, we are the
first to demonstrate such efficiency for batch joins and splits
on a sequence data structure supporting fast search. Our skip list data
structure can also be augmented to support efficient computation over
contiguous subsequences within the same efficiency bounds.

\noindent\textbf{A parallel Euler tour tree.} We
apply our skip lists to develop Euler tour trees that support parallel bulk
updates. Our Euler tour tree algorithms for adding
and for removing a batch of $k$ edges achieve $O(k \log (1 + n/k))$ expected
work and $O(\log n)$ depth with high probability.
These are the \emph{best known bounds for the batch-parallel dynamic trees
problem}.

\noindent\textbf{Experimental evidence of good performance.}
Our skip list and Euler tour tree data structures achieve good
self-relative speedups, ranging from $67\times$ to $96\times$ for large batch
sizes on 72 cores with
hyper-threading in our experiments. We also show that they
significantly outperform the fastest existing sequential alternatives.
Our code is publicly
available\footnote{\url{https://github.com/tomtseng/parallel-euler-tour-tree}}.

\section{Related Work}

\subsection{Sequences}

A common data structure for representing sequences is the search tree.
Concurrent binary search trees, however, tend to be hard to maintain because of
the frequent tree rebalancing necessary to preserve fast access times. Kung and
Lehman present a concurrent binary search tree supporting search, insertion, and
deletion~\cite{kung1980concurrent}. Their implementation grabs locks during
rebalancing, which blocks searches from proceeding. Ellen et al.\ provide a
lock-free binary search tree with the downside that the tree has no balance
guarantees~\cite{ellen2010non}. Braginsky and Petrank design a lock-free
balanced tree in the form of a B+ tree~\cite{braginsky2012lock}.

Batch parallelism for search, insertions, and deletions has been
studied in 2--3 trees~\cite{paul1983parallel}, red-black
trees~\cite{park2001parallel}, and
B-trees~\cite{higham1994maintaining}. All of these data structures
achieve $O(k \log n)$ work and $O(\log n + \log k)$ depth.

Very recently, Akhremtsev and Sanders implement parallel joins and
splits for $(a,b)$-trees as subroutines for efficient batch
updates~\cite{akhremtsev2016fast}. The work for batch joins is $O(k
\log (1 + n/k))$, and the work for batch splits is $O(k \log n)$. The
depth for both operations is $O(\log n)$.
Compared to~\cite{akhremtsev2016fast}, our skip lists are simpler,
allow augmentation, and improve on the work for batch splits. However,
as $(a,b)$-trees are a deterministic data structure, Akhremtsev and Sanders
obtain deterministic bounds whereas our bounds are randomized.

\myparagraph{Skip lists}
Skip lists are a randomized data structure introduced by Pugh for
representing ordered sequences~\cite{pugh1990skip}.  Concurrent skip
lists may be used as the basis for
dictionaries~\cite{sundell2004scalable} and priority
queues~\cite{linden2013skiplist,shavit2000skiplist,sundell2005fast}. Skip lists are also used for storing
database indices. For example, the popular database system MemSQL
builds its indices upon skip lists~\cite{memsqlindexdocs}. To the best
of our knowledge, no existing skip-list implementation supports batch-parallel
bounds for performing batches of splits or joins.

Pugh~\cite{pugh1990concurrent}, Herlihy et
al.~\cite{herlihy2006provably}, and Fraser~\cite{fraser2004practical} describe
concurrent skip lists supporting search, insertion, and deletion. They allow all
operations to run concurrently and do not show theoretical bounds. Gabarr\'{o}
et al.\ present a skip list supporting batch searches, insertions, or deletions in
$O(k (\log n + \log k))$ expected work and $O(\log n + \log k)$ expected
depth~\cite{gabarro1996design}.

\subsection{Dynamic Trees}
Many sequential data structures exist for the dynamic trees problem. Sleator and
Tarjan introduced the problem and gave a sequential data structure known as the
ST-tree or the link-cut tree for the problem~\cite{sleator1983data}. ST-trees
are difficult to parallelize because they rely on a complicated biased search
tree data structure.  Sleator and Tarjan later showed that ST-trees
can be significantly simplified by using splay trees~\cite{sleator1985self}. However, splay
trees are not amenable to parallelization due to the major structural
changes on every access caused by splaying nodes. Another data structure is Frederickson's topology
tree, which works by hierarchically clustering the represented
forest~\cite{frederickson1985data}. Acar et al.'s RC-trees similarly contracts
the forest to obtain a clustering~\cite{acar2004dynamizing}.
Unfortunately, both of these data structures require the forest to
have bounded degree and thus require modifying the original graph by
splitting high degree vertices into several bounded degree vertices.
Top trees, devised by Alstrup et al., circumvent this restriction and
allow for unbounded degree~\cite{alstrup2005maintaining}. They also
have the most general interface. The
Euler tour tree, developed by Miltersen et
al.~\cite{miltersen1994complexity}\ and Henzinger and
King~\cite{henzinger1995randomized}, is arguably the simplest data
structure for solving the dynamic trees problem,
but, unlike many other dynamic trees data structures, they do not support path queries.

Acar et al.\ very recently developed a batch-parallel solution to the dynamic
trees problem~\cite{acar2017brief}.  They achieve the same work bound
as our solution of $O(k
\log (1 + n/k))$ in expectation, but their depth bound is $O(C(k)\log n)$ where
$C(k)$ is the depth of compacting $k$ elements. As $C(k)$ is
$\Omega\left(\log^* k\right)$~\cite{mackenzie1992load}, our Euler tour trees
achieve better depth. Their data structures also require transforming
the input forest into a bounded-degree forest in order to use parallel
tree-contraction efficiently~\cite{miller1985parallel}.

\subsection{Parallel Dynamic Connectivity}
The dynamic trees problem with connectivity queries is the dynamic connectivity
problem restricted to acyclic graphs. A nearly ubiquitous strategy for dynamic
connectivity is to maintain a spanning forest of the graph as it undergoes
modifications. The difficulty of dynamic connectivity comes from discovering a
replacement edge going across a cut after deleting an edge in the maintained
spanning forest. Stipulating that the represented
graph is acyclic, as we do in this paper, simplifies the problem
because it guarantees that edge removal breaks a connected component
into two.

Though there is much work on sequential dynamic
connectivity~\cite{frederickson1985data, henzinger1995randomized,
thorup2000near, holm2001poly, kapron2013dynamic}, parallel dynamic
connectivity is not well explored. McColl et al.\ provide a parallel
algorithm for batch dynamic connectivity including edge deletions, but
their goal is to achieve fast experimental results on real-world
graphs rather than to achieve provable efficiency
bounds across all graphs~\cite{mccoll2013new}. The worst-case work and
depth of their algorithm is the same as that of a breadth-first search on the graph. Simsiri et
al.\ give a work-efficient, logarithmic-depth algorithm for batch
incremental (no deletions) dynamic connectivity~\cite{simsiri2016work}. Kopelowitz et
al.\  have recently shown that the sparsified version of
Frederickson's algorithm~\cite{frederickson1985data, eppstein1997sparsification}
can be parallelized nearly work-efficiently for a single
update~\cite{kopelowitz2018mst}.
However, they do not consider parallelizing the algorithm when
processing batches of edge updates.

\section{Preliminaries}\label{sec:prelims}
In this paper we analyze our algorithms on the
\defn{Multi-Threaded Random-Access Machine} (\mpram{}), a simple
work-depth model that is closely related to the Parallel Random-Access Machine (PRAM) but more closely
models current machines and programming paradigms that are
asynchronous and support dynamic forking. We define the
model in Appendix~\ref{sec:app-model} and refer the interested reader
to~\cite{blelloch2018introduction-draft} for more details.
Our efficiency bounds are stated in terms of work and depth, where
\defn{work} is the total number of vertices in the \process{} directed acyclic graph (DAG) and
where \defn{depth} (\defn{span}) is the length of the longest path in
the DAG~\cite{blelloch1996programming}.

We design algorithms using \emph{nested fork-join parallelism} in which a
procedure can \emph{fork} off another procedure call to run in parallel and
then wait for forked calls to complete with a \emph{join}
synchronization~\cite{blelloch1996programming}.  In our implementations, we use
Cilk Plus~\cite{leiserson2009cilk++} for fork-join parallelism. We
borrow its use of \emph{spawn} to mean ``fork'' and \emph{sync} to
mean ``join'' to disambiguate from the other sense in which we use
``join'' in this work (that is, as an operation that concatenates two
sequences).

Our algorithms only require the compare-and-swap atomic primitive,
which is widely available on modern multicores.  A
$\textproc{compare-and-swap}(\&x, o, n)$ (CAS) instruction takes a memory
location $x$ and atomically updates the value at location $x$ to $n$
if the value is currently $o$, returning $\codevar{true}$ if it
succeeds and $\codevar{false}$ otherwise.

\myparagraph{Parallel Primitives}
The following parallel procedures are used throughout the paper.

A \defn{semisort} takes an input array of
elements, where each element has an associated key, and reorders the
elements so that elements with equal keys are contiguous. Elements
with different keys are not necessarily ordered. The purpose is to
collect equal keys together rather than sort them. Semisorting a
sequence of length $n$ can be performed in $O(n)$ expected work and
$O(\log n)$ depth with high probability assuming access to a uniformly
random hash function mapping keys to integers in the range $\left[1,
n^{O(1)}\right]$~\cite{reif2000parallel, gu2015top}.

A \defn{parallel dictionary} data structure supports batch insertion,
batch deletion, and batch lookup of elements from some universe with
hashing.  Gil
et al.\ describe a parallel dictionary that
uses linear space and achieves $O(k)$ work and $O(\log^* k)$ depth with
high probability for a batch of $k$ operations~\cite{gil1991towards}.

The \defn{list tail-finding} problem is to assign each node in a linked
list a pointer to the last node in the list. There are
also other variants referred to as list ranking in the literature in
which we wish to compute the distance to the last node.
Many solutions for this problem that have $O(n)$ work and $O(\log n)$
depth exist~\cite{cole1988approximate, anderson1988deterministic,
cole1989faster, anderson1990simple}.

The \defn{pack} operation takes an $n$-length sequence $A$ and an
$n$-length sequence $B$ of booleans as input. The output is a sequence
$A'$ of all the elements $a \in A$ such that the the corresponding
entry in $B$ is $\codevar{true}$. The elements of $A'$ appear in the
same order that they appear in $A$. Packing can be easily implemented
in $O(n)$ work and $O(\log n)$ depth~\cite{jaja1992introduction}.

\section{Sequences and Parallel Skip Lists}\label{sec:skip-lists}
We start by first specifying a high-level interface for batch-parallel
sequences. We then describe our batch-parallel skip lists which
implement the interface, and finally, end by discussing how our data
structure can be extended to support augmentation.

\subsection{Batch-Parallel Sequence Interface}

The goal of a batch-parallel sequence data structure is to represent a
collection of sequences under batches of parallel operations that split and join
the sequences. To \emph{join} two sequences is to concatenate them together. To
\emph{split} a sequence $A$ at element $x$ is to separate the sequence into two
subsequences, the first of which consists of all elements in $A$ before and
including $x$, the second of which consists of all elements after $x$.

\myparagraph{Sequences}
We now give a formal description of the interface for sequences. The data
structure supports the following functions:
\begin{itemize}
  \item \textbf{$\textproc{BatchJoin}(\set{(x_1, y_1), \dots, (x_k, y_k)})$}
    takes an array of tuples where the $i$-th tuple is a pointer to the last
    element $x_i$ of one sequence and a pointer to the first element $y_i$ of a
    second sequence. For each tuple, the first sequence is concatenated with the
    second sequence. For any distinct tuples $(x_i, y_i)$ and $(x_j, y_j)$
    in the input, we must have $x_i \not= x_j$ and $y_i \not= y_j$.
  \item \textbf{$\textproc{BatchSplit}(\set{x_1, \ldots, x_k}$} takes an array of
    pointers to elements and, for each element $x_i$, breaks the sequence
    immediately after $x_i$.
  \item \textbf{$\textproc{BatchFindRep}(\set{x_1, \ldots, x_k}$} takes an array of
    pointers to elements. It returns an array where the $i$-th entry is the
    \emph{representative} of the sequence in which $x_i$ lives. The
    representative is defined so that $\on{representative}(u) =
    \on{representative}(v)$ if and only if $u$ and $v$ live in the same
    sequence. Representatives are invalidated after sequences are modified.
\end{itemize}

\myparagraph{Augmented Sequences}
To augment a sequence, we take an associative function $f : D^2 \to D$ where $D$
is an arbitrary domain of values. A value from $D$ is assigned to each element
in the sequence, $A$. An augmented sequence data structure supports querying for
the value of $f$ over contiguous subsequences of $A$. Specifically, our
interface for augmented sequences extends the interface for unaugmented
sequences with the following functions:
\begin{itemize}
  \item \textbf{$\textproc{BatchUpdateValue}(\set{(x_1, a_1), \dots, (x_k, a_k)})$}
    takes an array of tuples where the $i$-th tuple contains a pointer to an
    element $x_i$ and a new value $a_i \in D$ for the element. The value for $x_i$ is
    set to $a_i$ in the sequence.
  \item \textbf{$\textproc{BatchQueryValue}(\set{(x_1, y_1), \dots, (x_k, y_k)})$}
    takes an array of $k$ tuples where the $i$-th tuple contains pointers to
    elements $x_i$ and $y_i$. The return value is an array where the $i$-th
    entry holds the value of $f$ applied over the subsequence between $x_i$ and
    $y_i$ inclusive. For $1 \le i \le k$, $x_i$ and $y_i$ must be
    elements in the same sequence, and $y_i$ must appear after $x_i$ in the
    sequence.
\end{itemize}

\subsection{Skip Lists}
Skip lists are a simple randomized data structure that can be used to
represent sequences~\cite{pugh1990skip}.
To represent a sequence, skip lists assign a \emph{height} to each element of the
sequence, where each height is drawn independently from a geometric
distribution. The $\ell$-th level of a skip list consists of a linked list over
the subsequence formed by all elements of height at least $\ell$. This structure
allows efficient search. Figure~\ref{fig:sl-example} shows an example skip list.

\begin{figure} \centering
    \includegraphics[keepaspectratio=true, width=0.6\textwidth]{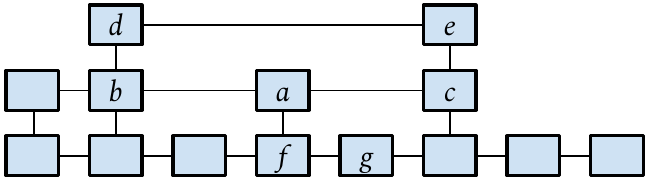}
    \caption[Example skip list] {An example skip list
    over a sequence of eight elements. On the bottom are all the level-$1$ nodes.}
    \label{fig:sl-example}
\end{figure}

For an element $x$ of height $h$, we allocate a node $v_i$ for every level
$i=1,2,\dots,h$. Each node has four pointers \textproc{left}, \textproc{right},
\textproc{up}, and \textproc{down}. We set $v_i \rightarrow \textproc{up} =
v_{i+1}$ and $v_i \rightarrow \textproc{down} = v_{i-1}$ for each $i$ to connect
between levels. We set $v_i \rightarrow \textproc{right}$ to the $i$-th node of
the next element of height at least $i$ and similarly $v_i \rightarrow
\textproc{left}$ to the $i$-th node of the previous element of height at least
$i$.

Our skip lists support \emph{cyclicity}, which is to say that our
algorithms are valid even if we link the tail and head of a skip list
together. Though this is not conventionally done with sequence data
structures, we will find it useful for representing Euler tours of
graphs in Section~\ref{sec:ett} since Euler tours are naturally cyclic
sequences. We cannot join upon cyclic sequences, but splitting a
cyclic sequence at element $x$ corresponds to unraveling it into a
linear sequence with its last element being $x$.
Figure~\ref{fig:sl-join-split} illustrates joining and splitting on
our skip lists.

\myparagraph{Definitions}
We now introduce definitions that describe the relationship between nodes. Say we
have a node $v$ that represents element $x$ at some level $i$. We call $v \rightarrow
\textproc{right}$ $v$'s \emph{successor}. Similarly,
$v \rightarrow \textproc{left}$ is its \emph{predecessor}. We call $v
\rightarrow \textproc{up}$ its \emph{direct parent} and $v \rightarrow
\textproc{down}$ its \emph{direct child}. For example, in
Figure~\ref{fig:sl-example}, consider node $a$. Its predecessor is $b$, its
successor is $c$, its direct child is $f$, and it has no direct parent.

The $\emph{left parent}$ is the level-$(i+1)$ node of the latest element
preceding and including $x$ that has height at least $i+1$. The $\emph{right
parent}$ is defined symmetrically. Under this definition, if $v$ has a direct
parent, then its left and right parents are both its direct parent. When we
refer to $v$'s parent, we refer to its left parent. In
Figure~\ref{fig:sl-example}, $a$'s (left) parent is $d$, and $a$'s right parent
is $e$. The \emph{(left) ancestors} consist of $v$'s parent, $v$'s parent's
parent, and so on, and similarly for $v$'s \emph{right ancestors}. Thus the
ancestors for both $f$ and $g$ in Figure~\ref{fig:sl-example} are $a$ and $d$.
A $\emph{child}$ is inverse to a parent, and a $\emph{descendant}$ is inverse to
an ancestor.

The following definitions describe the relationship between the links connecting
nodes.  The $\emph{parent}$ of a link between $v$ and its successor is the link
between $v$'s parent and its successor. Similarly, the $\emph{ancestors}$ of the
link are links between $v$'s ancestors and their successors. The \emph{children}
of the link are the links between $v$'s children and their successors.

\begin{figure}
    \centering
    \includegraphics[keepaspectratio=true, width=0.7\textwidth]{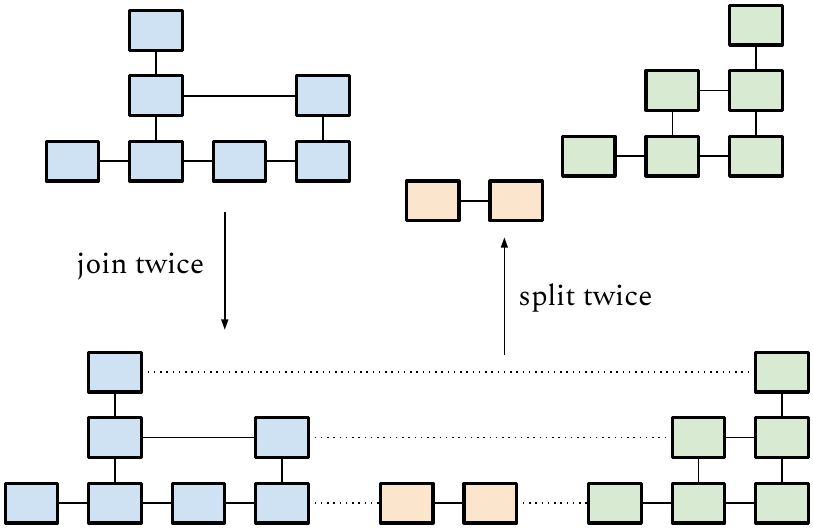}
    \caption[Joins and splits on skip lists] {Joins and
    splits on skip lists.}
    \label{fig:sl-join-split}
\end{figure}

\begin{figure}
    \centering
    \includegraphics[keepaspectratio=true, width=0.6\textwidth]{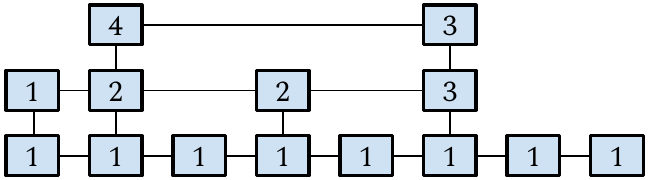}
    \caption[Skip list augmented with size] {Each node in this skip list is augmented with a size value, summing
    all the values of its children.}
    \label{fig:size-augment}
\end{figure}

\myparagraph{Joins, Splits, and Augmentation on Skip Lists}
Recall that in an augmented sequence, we take an associative function $f : D^2
\to D$ for some domain $D$. Each element in the sequence $A$ is
assigned some value from $D$. By storing these values in the bottom
level of our skip list and storing partial ``sums'' at higher levels, we can
compute $f$ over contiguous subsequences of $A$ in logarithmic time.
For instance, in Figure~\ref{fig:size-augment}, we assign the value
$1$ to every element and choose $f: \N^2 \to \N$ to be the sum
function. For each node $v$, we store the sum of the values of $v$'s
children. By looking at $O(\log n)$ nodes, we can then compute the
size of the  sequence. These augmented values are also easy to
maintain as the skip list undergoes joins and splits.

Our skip lists support batch joins, batch splits, batch point updates
of augmented values, and batch finding representatives in $O(k \log (1
+ n/k))$ work in expectation and $O(\log n)$ depth with high
probability, where $k$ is the batch size and $n$ is the number of
elements in the lists. We analyze efficiency in
Appendix~\ref{appendix:skip-list-efficiency}.

This improves on the $\Theta(k \log n)$ expected work bound achieved
by conventional sequential joins and splits on augmented skip lists.
Intuitively, the reason we can achieve improved work-bounds is that if
a node has many updated descendants, our algorithm updates its
augmented value only once rather than multiple times.

\subsection{Algorithms for unaugmented lists}\label{subsec:skip-list-description}

\begin{algorithm}
  \caption{Creates an element with height distributed according to
  $\on{Geometric}(1-p)$.} \label{alg:sl-create}
\small
\begin{algorithmic}[1]
\Procedure{CreateNode}{\null}
  \State \textbf{allocate} $\codevar{node}$
  \State With probability $p$:
  \Indent
    \State $\codevar{node} \rightarrow \textproc{up} = \Call{CreateNode}{\null}$
    \State $\codevar{node} \rightarrow \textproc{up} \rightarrow \textproc{down} = \codevar{node}$
  \EndIndent
  \State \Return $\codevar{node}$
\EndProcedure
\end{algorithmic}
\end{algorithm}

We begin by describing unaugmented skip lists.  For creating elements in our
skip list, we fix a probability $0 < p < 1$ representing the expected proportion
of nodes at a particular level that have a direct parent at the next level. We
generate heights of elements by allocating a node and giving each node a direct
parent with probability $p$ independently, as seen in
Algorithm~\ref{alg:sl-create}. This is equivalent to drawing heights from a
$\on{Geometric}(1-p)$ distribution.

\begin{algorithm}
\caption{Searches for the left parent of the input node. The mirror function
  \textproc{SearchRight} is defined symmetrically.}\label{alg:search-left}
\small
\begin{algorithmic}[1]
\Procedure{SearchLeft}{$\codevar{v}$}
  \State $\codevar{current} = \codevar{v}$
  \While {$\codevar{current} \rightarrow \textproc{up} = \codevar{null}$}
    \State $\codevar{current}$ = $\codevar{current} \rightarrow \textproc{left}$
      \label{line:search-update-current}
    \If {$\codevar{current} = \codevar{null}$ {\bf or} $\codevar{current} = \codevar{v}$}
      \label{line:search-conditional}
      \State \Return{$\codevar{null}$}
    \EndIf
  \EndWhile
  \State \Return $\codevar{current} \rightarrow \textproc{up}$
\EndProcedure
\end{algorithmic}
\end{algorithm}

We give pseudocode for \textproc{Join} and \textproc{Split} over
unaugmented lists in Algorithms~\ref{alg:join} and~\ref{alg:split}
respectively. To perform a batch of joins, we simply call
\textproc{Join} on each join operation in the batch concurrently, and
similarly for a batch of splits. As each batch of splits and joins
must be run in separate phases, our data structure is
\emph{phase-concurrent} over joins and splits, which is more general than being
batch-parallel~\cite{shun2014phase}.

Both algorithms employ two simple helper procedures,
\textproc{SearchLeft} and \textproc{SearchRight}, for finding the left
and right parents of a node. We show \textproc{SearchLeft} in
Algorithm~\ref{alg:search-left}, and \textproc{SearchRight} is
implemented symmetrically. Note that these procedures avoid looping
forever on cyclic skip lists.

\begin{algorithm}
\caption{Joins two lists together given their endpoints.} \label{alg:join}
\small
\begin{algorithmic}[1]
\Procedure{Join}{$\codevar{v}_L$, $\codevar{v}_R$}
 \If {$\Call{CAS}{\&\codevar{v}_L\rightarrow \textproc{right}, \codevar{null},
   \codevar{v}_R}$}
     \label{line:join-set-right}
   \State $\codevar{v}_R \rightarrow \textproc{left} = \codevar{v}_L$
     \label{line:join-set-left}
   \State $\codevar{parent}_{L}$ = \Call{SearchLeft}{$v_{L}$}
     \label{line:join-search-left}
   \State $\codevar{parent}_{R}$ = \Call{SearchRight}{$v_{R}$}
     \label{line:join-search-right}
   \If {$\codevar{parent}_{L} \neq \codevar{null}$ {\bf and } $\codevar{parent}_{R} \neq \codevar{null}$}
     \label{line:join-if-search}
     \State \Call{Join}{$\codevar{parent}_{L}, \codevar{parent}_{R}$}
       \label{line:join-recurse}
   \EndIf
 \EndIf
\EndProcedure
\end{algorithmic}
\end{algorithm}

\myparagraph{Join}
Recall that the definition of \textproc{Join} takes a pointer to the
last element of one list and a pointer to the first element of a
second list and  concatenates the first list with the second list.
Starting at the bottom level, our algorithm links the given nodes,
searches upwards to find parents to link at the next level, and
repeats. We set the link with a CAS, and if the CAS is lost, the
algorithm quits. This permits only one thread to set a particular
link, preventing repeated work.

\begin{theorem}\label{thm:join-correctness}
  Let $B$ be a set of valid \textproc{Join} inputs. Then calling
  \textproc{Join} concurrently over the inputs in $B$ gives the same result as
  joining over the inputs in $B$ sequentially.
\end{theorem}
\begin{proof}
(Proof sketch) We argue inductively level-by-level that all
necessary links are added and no unnecessary links are added. For the
base case, at the bottom level, the links we add are exactly those
given as input to the algorithm, which are the necessary links to add
at that level. For the inductive step, assume that the correct links
will be added on level $i$.  Consider any link $\ell$ from nodes $v_L$
to $v_R$ on level $i + 1$ that should be added. In order for this to
be a link we need to add, there must be a rightward path from $v_L$'s
direct child to $v_R$'s direct child once all links on level $i$ are
added. Then consider the last execution of \textproc{join} on level
$i$ to add a link on that path by finishing
line~\ref{line:join-set-left} of Algorithm~\ref{alg:join}. That
execution will have a complete path to find parents $v_L$ and $v_R$
when searching and thus will find $\ell$ as a link to add.
Conversely, any level-$(i+1)$ link from nodes $v_L$ to $v_R$ found by
a join execution was found via a complete path (albeit perhaps
temporarily missing some $\textproc{left}$ pointers due to some
executions of join completing line~\ref{line:join-set-right} but not
yet completing line~\ref{line:join-set-left}) between $v_L$'s direct
child and $v_R$'s direct child, which indicates that this link should
be added. We present a formal proof using this idea in
Appendix~\ref{appendix:join-correctness}.
\end{proof}

\begin{algorithm}
\caption{Separates the input node from its successor.} \label{alg:split}
\small
\begin{algorithmic}[1]
\Procedure{Split}{$\codevar{v}$} \label{line:split-begin}
\State $\codevar{ngh} = \codevar{v}\rightarrow \textproc{right}$
\If {$\codevar{ngh} \neq \codevar{null}$ {\bf and}
  $\Call{CAS}{\&\codevar{v}\rightarrow \textproc{right}, \codevar{ngh},
  \codevar{null}}$} \label{line:split-set-right}
  \State $\codevar{ngh}\rightarrow \textproc{left} = \codevar{null}$
    \label{line:split-set-left}
  \State $\codevar{parent} = \Call{SearchLeft}{v}$
    \label{line:split-search}
  \If {$\codevar{parent} \neq \codevar{null}$}
    \State $\Call{Split}{\codevar{parent}}$
    \label{line:split-recurse}
  \EndIf
\EndIf
\EndProcedure
\end{algorithmic}
\end{algorithm}

\myparagraph{Split}
\textproc{Split} takes a pointer to an element and breaks the list right
after that element. Similar to join, it cuts the link at the bottom level and
then loops in searching upwards to find parent links to remove at higher
levels. Like \textproc{Join}, this uses CAS to avoid duplicate work.

\begin{theorem}\label{thm:split-correctness}
  Let $B$ be a set of elements. Then calling \textproc{Split} concurrently over
  the elements in $B$ gives the same result as splitting over the elements in
  $B$ sequentially.
\end{theorem}
\begin{proof}
(Proof sketch) Like in the proof sketch of Theorem~\ref{thm:join-correctness},
we look at the links that are removed inductively level-by-level. The argument
is similar, except that in the inductive step, to see that a link $\ell$ on
level $i + 1$ from nodes $\codevar{v}_L$ to $\codevar{v}_R$ that should be
removed will indeed be removed by the phase of splits, we note that the
leftmost split on the path from $\codevar{v}_L \rightarrow \textproc{down}$ to
$\codevar{v}_R \rightarrow \textproc{down}$ will be able to find parent
$\codevar{v}_L$ in its $\textproc{SearchLeft}$ call. We present a formal proof
using this idea in Appendix~\ref{appendix:split-correctness}.
\end{proof}

\myparagraph{Finding representative nodes}
\begin{algorithm}
\caption{Finds a representative node of the list that the input node lives in.} \label{alg:find-rep}
\small
\begin{algorithmic}[1]
\Procedure{FindRep}{$\codevar{v}$}
  \While {$\textproc{SearchRight}(\codevar{v}) \not= \codevar{null}$}
    \label{line:find-rep-search-right-loop}
    \State $\codevar{v} = \textproc{SearchRight}(\codevar{v})$
  \EndWhile
    \label{line:find-rep-search-right-loop-end}
  \While {$\textproc{SearchLeft}(\codevar{v}) \not= \codevar{null}$}
    \State $\codevar{v} = \textproc{SearchLeft}(\codevar{v})$
  \EndWhile
  \State $\codevar{rep} = \codevar{v}$
  \While {true}
    \If {$\codevar{v} \rightarrow \textproc{left} = \codevar{null}$} \Comment{List is acyclic}
      \State \Return $\codevar{v}$
    \EndIf
    \State $\codevar{v} = \codevar{v} \rightarrow \textproc{left}$
    \If {$\codevar{v} = \codevar{rep}$} \Comment{List is cyclic}
      \State \Return $\codevar{rep}$
    \EndIf
    \If {$\codevar{v} < \codevar{rep}$}
      \State $\codevar{rep} = \codevar{v}$
    \EndIf
  \EndWhile
\EndProcedure
\end{algorithmic}
\end{algorithm}

A simple phase-concurrent implementation of \textproc{FindRep} that takes $O(k
\log n)$ expected work for $k$ concurrent calls is to start at the
input node and walk to the top level of the list. Then on the top
level, for an acyclic list, we return the leftmost node, or for a
cyclic list, we return the node with the lowest memory address. This is shown in
Algorithm~\ref{alg:find-rep}.

However, if we are given a batch of $k$ calls up front, we can in fact achieve
$O(k \log (1 + n/k))$ expected work and $O(\log n)$ depth with high probability.
The idea is that each call of \textproc{FindRep} takes some path up the skip
list to the top level, and calls whose paths intersect somewhere can be combined
at that point to avoid duplicate work. Then the return value gets propagated
back down to both original calls. The code would look similar to the code for
batch updating augmented values for augmented skip lists
(Subsection~\ref{subsec:augment-skip-lists}) in
Algorithm~\ref{alg:batch-update}. We omit the full details.

\subsection{Algorithm for augmented lists}\label{subsec:augment-skip-lists}

We now describe how to augment our skip lists. In addition to its four
pointers, each node is given a value \textproc{val} from some domain
$D$ and a boolean \textproc{needs\_update}. We provide an
associative function $f: D^2 \to D$ and, for each element in the list, a value
from $D$. We assign values to $\textproc{val}$ on nodes at the bottom
level and then compute \textproc{val} at higher levels by applying $f$ over
nodes' children. The boolean \textproc{needs\_update} is initialized to
$\codevar{false}$ and is used to mark nodes whose values need updating.

\begin{algorithm}
  \caption{Helper function for \textproc{BatchUpdateValues} that updates the
  augmented value for $v$ and all its descendants.}
  \label{alg:batch-update-helper-top-down}
\small
\begin{algorithmic}[1]
  \Procedure{UpdateTopDown}{$v$}
    \State $v \rightarrow \textproc{needs\_update} = \codevar{false}$
    \If {$v \rightarrow \textproc{down} = \codevar{null}$} \Comment Reached bottom level
      \State \Return
    \EndIf
    \State $\codevar{current} = v \rightarrow \textproc{down}$
      \label{line:updatetopdown-child-update-start}
    \Do
      \If{$\codevar{current} \rightarrow \textproc{needs\_update}$}
        \State \textbf{spawn} $\Call{UpdateTopDown}{current}$
      \EndIf
      \State $\codevar{current} = \codevar{current} \rightarrow \textproc{right}$
    \doWhile{$\codevar{current} \not= \codevar{null}$ \textbf{and}
             $\codevar{current} \rightarrow \textproc{up} = \codevar{null}$}
      \label{line:updatetopdown-child-update-end}
    \State \textbf{sync}
      \label{line:updatetopdown-self-update-start}
    \State $\codevar{sum} = v \rightarrow \textproc{down} \rightarrow \textproc{val}$
    \State $\codevar{current} = v \rightarrow \textproc{down} \rightarrow \textproc{right}$
    \While{$\codevar{current} \not= \codevar{null}$ \textbf{and}
             $\codevar{current} \rightarrow \textproc{up} = \codevar{null}$}
      \State $\codevar{sum} = f(\codevar{sum}, \codevar{current} \rightarrow \textproc{val})$
      \State $\codevar{current} = \codevar{current} \rightarrow \textproc{right}$
    \EndWhile
    \State $v \rightarrow \textproc{val} = \codevar{sum}$
      \label{line:updatetopdown-self-update-end}
  \EndProcedure
\end{algorithmic}
\end{algorithm}

\begin{algorithm}
  \caption{Takes a batch of (node, value) pairs, updates each node with its
  associated value, and updates other affected augmented values stored throughout the list.}
  \label{alg:batch-update}
\small
\begin{algorithmic}[1]
  \Procedure{BatchUpdateValues}{$\set{(v_1, a_1), \dots, (v_k, a_k)}$}
    \State $\codevar{top} = \set{\codevar{null}, \codevar{null}, \dots, \codevar{null}}$ \Comment{$k$-length array}
    \label{line:batch-update-cas-phase-start}
    \For{$i \in \set{1, \dots, k}$} \textbf{in parallel}
      \State $v_i \rightarrow \textproc{val} = a_i$
      \State $\codevar{current} = v_i$
        \label{line:batch-update-update-bottom}
      \While {$\Call{CAS}{\&\codevar{current}\rightarrow\textproc{needs\_update}, \codevar{false}, \codevar{true}}$}
        \label{line:batch-update-cas}
        \State $\codevar{parent} = \Call{SearchLeft}{current}$
        \If{$\codevar{parent} = \codevar{null}$}
          \State $\codevar{top}[i] = current$
          \State \textbf{break}
        \EndIf
        \State $\codevar{current} = \codevar{parent}$
      \EndWhile
    \EndFor
    \label{line:batch-update-cas-phase-end}
    \For{$i \in \set{1, \dots, k}$} \textbf{in parallel}
      \label{line:batch-update-helper-start}
      \If{$\codevar{top}[i] \not= \codevar{null}$}
        \State $\Call{UpdateTopDown}{\codevar{top}[i]}$
      \EndIf
    \EndFor
      \label{line:batch-update-helper-end}
  \EndProcedure
\end{algorithmic}
\end{algorithm}

We give the main algorithm \textproc{BatchUpdateValues} for batch augmented value update in
Algorithm~\ref{alg:batch-update}. This takes a set of nodes at the bottom
level along with values to give to the associated elements.
For each node in the set, we start by updating its value
(line~\ref{line:batch-update-update-bottom}). Then each of its
ancestors have values that need updating, so we walk up its ancestors, CASing on
each ancestor's \textproc{needs\_update} variable (line~\ref{line:batch-update-cas}). If an execution loses a
CAS, then it may quit because some other execution will take care of all the
node's ancestors.

Now over all the input nodes that won all CASes on their ancestors, we know the
union of their topmost ancestors' descendants contain all the input nodes.
By calling the helper function \textproc{UpdateTopDown}
(Algorithm~\ref{alg:batch-update-helper-top-down}) on every such topmost
ancestor in
lines~\ref{line:batch-update-helper-start}--\ref{line:batch-update-helper-end},
we traverse back down and update these descendants' augmented values. Given a
node, this helper function calls itself recursively on all the node's children
$c$ who need an update as indicated by $c \rightarrow \textproc{needs\_update}$
(lines~\ref{line:updatetopdown-child-update-start}--\ref{line:updatetopdown-child-update-end}).
Then, after all the childrens' values are updated, we may update the original
node's value
(lines~\ref{line:updatetopdown-self-update-start}--\ref{line:updatetopdown-self-update-end}).

\begin{algorithm}
  \caption{Batch join for augmented skip lists.}
  \label{alg:batch-join}
\small
\begin{algorithmic}[1]
  \Procedure{BatchJoin}{$\set{(l_1, r_1), (l_2, r_2), \dots, (l_k, r_k)}$}
    \For{$i \in \set{1, \dots, k}$} \textbf{in parallel}
      \State $\Call{Join}{l_i, r_i}$
    \EndFor
    \State $\Call{BatchUpdateValues}{\set{(l_1, l_1 \rightarrow
      \textproc{val}), \dots, (l_k, l_k \rightarrow \textproc{val})}}$
  \EndProcedure
\end{algorithmic}
\end{algorithm}

\begin{algorithm}
  \caption{Batch split for augmented skip lists.}
  \label{alg:batch-split}
\small
\begin{algorithmic}[1]
  \Procedure{BatchSplit}{$\set{v_1, v_2, \dots, v_k}$}
    \For{$i \in \set{1, \dots, k}$} \textbf{in parallel}
      \State $\Call{Split}{v_i}$
    \EndFor
    \State $\Call{BatchUpdateValues}{\set{(v_1, v_1 \rightarrow
      \textproc{val}), \dots, (v_k, v_k \rightarrow \textproc{val})}}$
  \EndProcedure
\end{algorithmic}
\end{algorithm}

With this algorithm for batch augmented value update, batch joins
(Algorithm~\ref{alg:batch-join}) and batch splits
(Algorithm~\ref{alg:batch-split}) are simple. We first perform all the joins or
splits. Then we batch update on the nodes we joined or split on. We keep all the
values on the bottom level the same, but the update fixes all the values on the
higher levels that are changed by adding or removing links.

\begin{algorithm}
  \caption{Query for the augmented value over the subsequence between $v_L$ and
  $v_R$ inclusive.  Here we let $f(\codevar{none}, x) = x$ for all $x \in D$.}
  \label{alg:aug-query-value}
\small
\begin{algorithmic}[1]
  \Procedure{QueryValue}{$v_L$, $v_R$}
  \State $\codevar{sum}_L = \codevar{none}$
  \State $\codevar{sum}_R = v_R \rightarrow \textproc{val}$
  \While{$v_L \not= v_R$}
    \While{$v_L \rightarrow \textproc{up} \not= \codevar{null}$ \textbf{and}
           $v_R \rightarrow \textproc{up} \not= \codevar{null}$}
      \State $v_L = v_L \rightarrow \textproc{up}$
      \State $v_R = v_R \rightarrow \textproc{up}$
    \EndWhile
    \If {$v_L \rightarrow \textproc{up} = \codevar{null}$}
      \State $\codevar{sum}_L = f(\codevar{sum}_L, v_L \rightarrow \textproc{val})$
      \State $v_L = v_L \rightarrow \textproc{right}$
    \Else
      \State $v_R = v_R \rightarrow \textproc{left}$
      \State $\codevar{sum}_R = f(v_R \rightarrow \textproc{val}, \codevar{sum}_R)$
    \EndIf
  \EndWhile
  \Return $f(\codevar{sum}_L, \codevar{sum}_R)$
  \EndProcedure
\end{algorithmic}
\end{algorithm}

A batch of $k$ queries for the augmented values over contiguous subsequences of
lists can be processed in $O(k \log n)$ expected work and $O(\log n)$ depth
with high probability by simply performing each query with
Algorithm~\ref{alg:aug-query-value} in parallel at $O(\log n)$ expected work
per query.

\subsection{Implementation}\label{subsec:skip-list-implement}
We provide details about our skip list implementations in
Appendix~\ref{sec:app-skip-list-implementation}.

\section{Batch-Parallel Euler Tour Trees}\label{sec:ett}
In this section we present batch-parallel Euler tour trees, a solution
to the batch-parallel dynamic trees problem. In order to ease
exposition, we first present a batch-parallel interface for
the dynamic trees problem.

\myparagraph{Batch-Parallel Dynamic Trees Interface}
A solution to the batch-parallel dynamic trees problem supports
representing a forest as it undergoes batches of links, cuts, and
connectivity queries. \emph{Links} link two trees in the forest.
\emph{Cuts} delete an edge from the forest and break one tree into two
trees. \emph{Connectivity} queries take two vertices in the forest and
return whether they are connected (that is, whether they are in the same tree). We now give a
formal description of the interface. The data structure maintains a
graph $G=(E, V)$, which is assumed to be a forest under the following
operations:
\begin{itemize}
  \item \textbf{$\textproc{BatchLink}(\set{\set{u_1, v_1}, \ldots, \set{u_k,
      v_k}})$} takes an array of edges and adds them to the graph
    $G$. The input edges must not create a cycle in $G$.

  \item \textbf{$\textproc{BatchCut}(\set{\set{u_1, v_1}, \ldots, \set{u_k,
      v_k}})$} takes an array of edges and removes them from the
    graph $G$.

  \item \textbf{$\textproc{BatchConnected}(\set{\set{u_1,v_1}, \ldots, \set{u_k,
      v_k}})$} takes an array of tuples representing queries. The
    output is an array where the $i$-th entry returns whether vertices
    $u_i$ and $v_i$ are connected by a path in $G$.
\end{itemize}

We also support augmenting the trees with an associative and
commutative function $f: D^2 \to D$ with values from $D$ assigned to
vertices and edges of the forest. The goal of augmentation is to compute $f$ over
subtrees of the represented forest. Note that we assume that the
function is commutative in order to allow implementations to not maintain any
specific order over each vertex's children. The interface
supports batch updates over vertices and edges. The primitives are
similar to the batch updates of values for augmented skip lists, so we
elide the details. The interface for subtree queries is different, and
we present it below:
\begin{itemize}
  \item \textbf{$\textproc{BatchSubtree}(\set{(u_1, p_1), \ldots, (u_k,
      p_k)})$} takes an array of tuples, where the $i$-th tuple
    contains a vertex $u_i$ and its parent $p_i$ in the tree. It
    returns an array where the $i$-th entry contains the value of $f$
    summed over $u_i$'s subtree relative to its parent $p_i$ in $G$.
    Note that because the represented trees are unrooted, we require
    providing the parent $p_i$ in order to determine the intended
    subtree for $u_i$.
\end{itemize}

(Some dynamic trees data structures allow queries for augmented values summed
over \emph{paths} rather than subtrees in the represented forest, but Euler tour
trees do not.)

\myparagraph{Euler Tour Trees} We focus on a variant of Euler tour trees presented by
Tarjan~\cite{tarjan1997dynamic}.  To represent a tree as an Euler tour
tree, replace each edge $\set{u, v}$ with two directed edges $(u, v)$
and $(v, u)$ and add a loop $(v, v)$ to each vertex $v$, as shown in
Figure~\ref{fig:et-transform}. This construction produces a
connected graph in which each vertex has equal indegree and outdegree,
and therefore the graph admits an Euler tour. We represent the tree as
any of its Euler tours.

\begin{figure}
    \centering
    \includegraphics[keepaspectratio=true, width=0.6\textwidth]{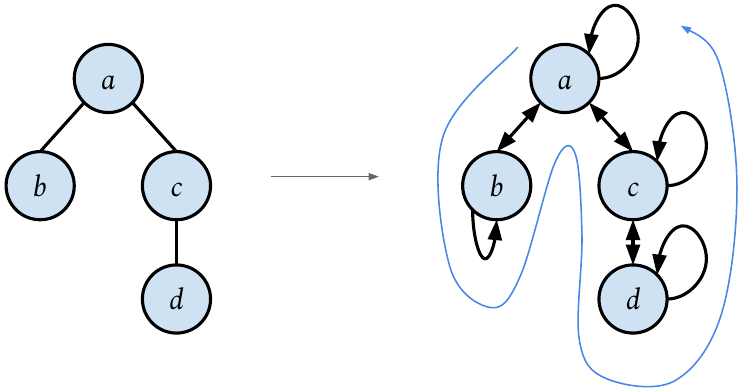}
    \caption[Transformation of a tree to obtain an Euler tour] {We take the
    tree on the left and transform it so that we get the following Euler tour of
    edges: $(a,b)$ $(b,b)$ $(b,a)$ $(a,c)$ $(c,d)$ $(d,d)$ $(d,c)$ $(c,c)$
    $(c,a)$ $(a,a)$.}
    \label{fig:et-transform}
\end{figure}

Now linking two trees corresponds to splicing their Euler tours together, and
cutting a tree corresponds to cutting out part of its Euler tour. Each of these
operations reduces to a few joins and splits on the tours. We may also
answer whether two vertices $u$ and $v$ are connected by asking whether their
loops, $(u, u)$ and $(v,v)$, reside in the same tour. Moreover, because a
subtree appears as a contiguous section of an Euler tour, we can efficiently
compute information about subtrees if we can efficiently compute information
about contiguous sections of tours.

Traditionally, Euler tour trees store an Euler tour by
breaking it into a sequence at an arbitrary location and then placing the
sequence in a balanced binary tree. We instead store Euler tours as cycles using
our skip lists from Section~\ref{sec:skip-lists}. Because skip lists are easy to
join and split in parallel, we can process batches of links and cuts on Euler
tour trees efficiently.

We show that for a batch of $k$ joins, $k$ splits, or $k$ connectivity queries
over an $n$-vertex forest, we can achieve $O(k \log (1 + n/k))$
expected work and $O(\log n)$ depth with high probability.  If we
build our Euler tour trees over augmented skip lists, we can also
answer subtree queries efficiently.

\subsection{Description}

Our Euler tour trees crucially rely on our parallel skip lists to represent
Euler tours. Since a graph of $n$ vertices has Euler tours whose lengths sum to
$O(n)$, the skip lists hold $O(n)$ elements. Thus a batch of $k$
joins or splits on the Euler tours takes $O(k \log (1 + n/k))$ expected work and
$O(\log n)$ depth with high probability.

\myparagraph{Construction}
For clarity, we describe our Euler tour trees using the
phase-concurrent unaugmented skip lists given in Section~\ref{sec:skip-lists}.
However, it is easy to organize the joins and splits into batches so as to
match the augmented skip list interface seen in
Subsection~\ref{subsec:augment-skip-lists}. We also treat our dictionary
data structure as phase-concurrent for clarity, but again, this is easy to
circumvent.

We add fields \textproc{twin} and \textproc{mark} to each
skip list element. For an element representing a directed edge $(u,v)$,
$\textproc{twin}$ is a pointer to the element representing the directed edge $(v,
u)$ in the opposite direction. We initialize the field \textproc{mark} to
$\codevar{false}$ and use it during splitting to mark elements that will be
removed.

\begin{algorithm}
  \caption{Euler tour tree data structure initialization.}
  \label{alg:et-init}
\small
\begin{algorithmic}[1]
  \Procedure{Initialize}{$n$}
    \State $\codevar{verts} = \set{}$ \Comment $n$-length array
    \For{$i \in \set{1, \dots, n}$} \textbf{in parallel}
      \State $\codevar{verts[i]} = \Call{CreateNode}{\null}$
      \State \Call{Join}{$\codevar{verts}[i], \codevar{verts}[i]$}
    \EndFor
    \State $\codevar{edges} = \Call{dict}{\null}$ \Comment empty dictionary
    \State $\codevar{successors} = \set{}$ \Comment $n$-length array
  \EndProcedure
\end{algorithmic}
\end{algorithm}

At initialization (Algorithm~\ref{alg:et-init}), the represented graph is an
$n$-vertex forest with no edges, and we assume the vertices are labeled with
integers $1, 2, \dots, n$. We create an $n$-length array $\codevar{verts}$ such
that $\codevar{verts}[i]$ stores a pointer to the skip list element representing
the loop edge $(i, i)$.  As such, in parallel, for $i = 1,\dots, n$, we create a
skip list element, assign it to $\codevar{verts}[i]$, and join it to itself to form
a singleton cycle.  These cycles are the Euler tours in an empty graph. We also
keep a dictionary $\codevar{edges}$ that maps edges $(u, v)$ with $u \not= v$ to
corresponding skip list elements. Lastly, we create an array
$\codevar{successors}$ that will be used as scratch space for batch linking.

\myparagraph{Connectivity queries}
To check whether two vertices are connected, we simply check whether they live
in the same Euler tour by comparing the representatives of their tours' skip
lists. The complexity of this can be made $O(k \log (1 + n/k))$ expected work
and $O(\log n)$ depth with high probability using an efficient
\textproc{BatchFindRep} algorithm.

\begin{algorithm}
  \caption{Add a batch of edges to Euler tour tree.}
  \label{alg:et-batch-link}
\small
\begin{algorithmic}[1]
  \Procedure{BatchLink}{$\set{\set{u_1, v_1}, \set{u_2, v_2}, \dots, \set{u_k, v_k}}$}
     \State \Comment Adding input edges must not create a cycle in graph
     \State \Comment Create nodes representing new edges
     \label{line:batch-link-new-nodes-start}
     \For{$i \in \set{1, \dots, k}$} \textbf{in parallel}
       \State $\codevar{uv} = \Call{CreateNode}{\null}$
       \State $\codevar{vu} = \Call{CreateNode}{\null}$
       \State $\codevar{uv} \rightarrow \textproc{twin} = \codevar{vu}$
       \State $\codevar{vu} \rightarrow \textproc{twin} = \codevar{uv}$
       \State $\codevar{edges}[(u_i, v_i)] = uv$
       \State $\codevar{edges}[(v_i, u_i)] = vu$
     \EndFor
     \label{line:batch-link-new-nodes-end}
     \State \Comment Cut at locations at which we splice in other tours
     \label{line:batch-link-split-start}
     \For{$i \in \set{1, \dots, k}$} \textbf{in parallel}
       \For{$w \in \set{u_i, v_i}$}
         \State $\codevar{w\_node} = \codevar{verts}[w]$
         \State $\codevar{w\_succ} = \codevar{w\_node} \rightarrow \textproc{right}$
         \If{$\codevar{w\_succ} \not= \codevar{null}$}
           \State \Comment \parbox[t]{.75\linewidth}{benign race; this assignment and split are idempotent}
           \State $\codevar{successors}[w] = \codevar{w\_succ}$
           \State \Call{Split}{$\codevar{w\_node}$}
         \EndIf
       \EndFor
     \EndFor
     \label{line:batch-link-split-end}
     \State $\codevar{sorted\_edges} =
       \newline\makebox[3em]{}
       \Call{Semisort}{
       \set{(u_1, v_1), (v_1, u_1), \dots, (u_k, v_k), (v_k, u_k)}}$
     \label{line:batch-link-semisort}
     \State \Comment Join together tours with new edge nodes in between
     \label{line:batch-link-join-start}
     \For{$i \in \set{1, \dots, 2k}$} \textbf{in parallel}
       \State $(\codevar{u}, \codevar{v}) = \codevar{sorted\_edges}[i]$
       \State $(\codevar{u\_prev}, \codevar{v\_prev}) = \codevar{sorted\_edges}[i - 1]$
       \State $(\codevar{u\_next}, \codevar{v\_next}) = \codevar{sorted\_edges}[i + 1]$
       \If{$i = 1$ \textbf{or} $u \not= u\_prev$}
         \State \Call{Join}{$(\codevar{verts}[\codevar{u}], \codevar{edges}[(\codevar{u}, \codevar{v})])$}
       \EndIf
       \If{$i = 2k$ \textbf{or} $\codevar{u} \not= \codevar{u\_next}$}
         \State \Call{Join}{$(\codevar{edges}[(\codevar{v}, \codevar{u})], \codevar{successors}[\codevar{u}])$}
       \Else
         \State \Call{Join}{$(\codevar{edges}[(\codevar{v}, \codevar{u})],
           \codevar{edges}[(\codevar{u\_next}, \codevar{v\_next})])$}
       \EndIf
     \EndFor
     \label{line:batch-link-join-end}
  \EndProcedure
\end{algorithmic}
\end{algorithm}

\myparagraph{Batch Link}
Algorithm~\ref{alg:et-batch-link} shows our algorithm for adding a batch of
edges. The algorithm takes an array of edges $A$ to add as input.
We assume that adding the input edges preserves acyclity.

To add a single edge $\set{u, v}$ sequentially, we can find
locations where $u$ and $v$ appear in their tours by looking up
$\codevar{verts}[u]$ and $\codevar{verts}[v]$. We split on those locations and
join the resulting cut up tours back together with new nodes representing
$(u,v)$ and $(v, u)$ in between. If we want to add several edges in parallel, we
need to be careful when inserting edges that are incident to the same vertex and
thus attempt to join on the same location.

With that in mind, we proceed to describe our algorithm. In lines
\ref{line:batch-link-new-nodes-start}-\ref{line:batch-link-new-nodes-end}, for
each input edge $\set{u, v}$, we allocate new list elements representing directed
edges $(u, v)$ and $(v, u)$. Then, in lines
\ref{line:batch-link-split-start}-\ref{line:batch-link-split-end}, for each
vertex $u$ that appears in the input, we split $u$'s list at
$\codevar{verts}[u]$ as a location to splice in other tours. We also
save the successor of $\codevar{verts}[u]$ in $\codevar{successors}[u]$ so that
we can join everything back together at the end.

For each vertex $u$, say that the input tells us that we want to newly connect
$u$ to vertices $w_1, w_2, \dots, w_k$.
Then we join together the nodes representing $(u, u)$ to $(u, w_1)$, $(w_i, u)$
to $(u, w_{i+1})$ for $1 \leq i < k$, and $(w_k, u)$ to what was the successor to $(u,
u)$ before splitting. In our code, we arrange this in lines
\ref{line:batch-link-semisort}-\ref{line:batch-link-join-end} by semisorting
the input to collect together all edges incident on a vertex. The ordering of
$w_1, w_2, \dots, w_k$ is unimportant, only corresponding to the order in
which they appear after $u$ in the Euler tour.

\begin{figure}
     \centering
     \includegraphics[keepaspectratio=true, width=0.25\textwidth]{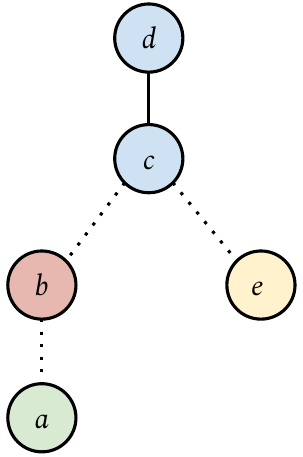}
     \caption[Example graph for Euler tour batch link] {An example graph for
     illustrating batch linking . The dashed edges $\set{a, b}, \set{b, c},
     \set{c,e}$ are new edges to add in a batch, whereas the solid edge is
     an already existing edge.}
     \label{fig:et-link-example-graph}
\end{figure}
\begin{figure}
    \centering
    \includegraphics[keepaspectratio=true, width=0.5\textwidth]{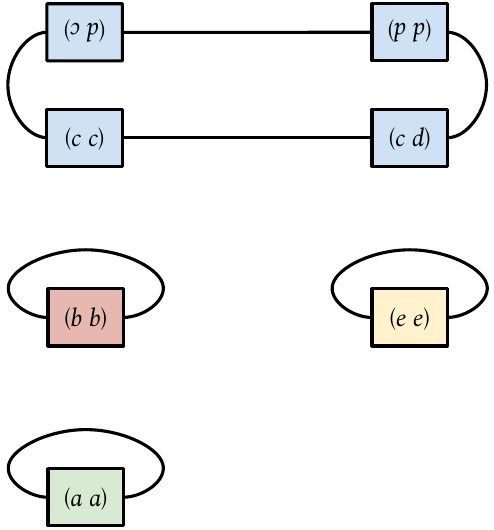}
    \caption[Euler tours before batch link] {We represent the
    graph from Figure~\ref{fig:et-link-example-graph} prior to adding the
    dashed edges with Euler tours stored in cyclic linked lists.}
    \label{fig:et-link-list-before}
\end{figure}
\begin{figure}
    \centering
    \includegraphics[keepaspectratio=true, width=0.75\textwidth]{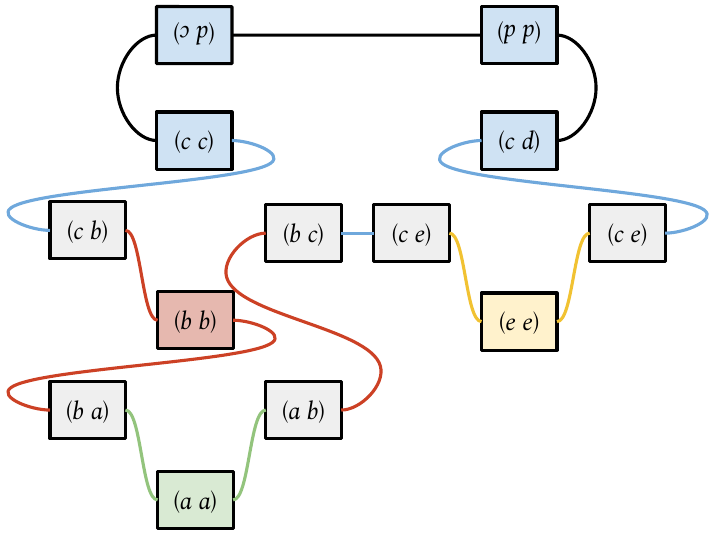}
    \caption[Euler tours after batch link] {After adding the dashed edges in
    figure~\ref{fig:et-link-example-graph}, the Euler tour stored in a linked
    list may look like this. In this figure, we color links in the list the same
    colors as the nodes ``responsible'' for adding those links in through calls
    to $\textproc{Join}$.}
    \label{fig:et-link-list-after}
\end{figure}

As an example of the desired result from a batch link, consider the graph in
Figure~\ref{fig:et-link-example-graph} on which we wish to perform a batch link
with input $\set{\set{a, b}, \set{b, c}, \set{c,e}}$. Prior to the batch link,
the Euler tour may look like Figure~\ref{fig:et-link-list-before}, and after the
batch link, the Euler tour may look like Figure~\ref{fig:et-link-list-after}.
Let us focus on what happens to the list containing vertex $c$. After allocating
the new nodes representing the new edges to add, we split node $(c, c)$ from its
successor $(c, d)$. We want to add edges connecting $c$ to $b$ and to $e$. As
such, we perform joins adding links from nodes $(c, c)$ to $(c, b)$, $(b, c)$ to
$(c, e)$, and $(e, c)$ to $(c, d)$.  These new links correspond to the blue
links in figure~\ref{fig:et-link-list-after}. Doing this over all vertices
provides all the joins needed to form an Euler tour over the whole graph.

Using our skip lists and an efficient semisort~\cite{gu2015top}, we see that the work is
$O(k\log(1 + n/k))$ in expectation, and the depth is $O(\log n)$ with high
probability.

\begin{algorithm}
  \caption{Computes locations at which to join for batch cut.}
  \label{alg:list-rank}
  \small
\begin{algorithmic}[1]
  \Procedure{GetNextUnmarked}{$\codevar{elements} = \set{z_1, z_2, \dots, z_k}$}
     \State \Comment Input is a set of skip list elements
  \For{$i \in \set{1, \dots, k}$} \textbf{in parallel}
    \State $\codevar{next} = z_i \rightarrow \textproc{twin} \rightarrow \textproc{right}$
    \If{$\codevar{next} \rightarrow \textproc{mark}$}
      \State $z_i \rightarrow \textproc{next\_edge} = \codevar{next}$
    \Else
      \State $z_i \rightarrow \textproc{next\_edge} = \codevar{null}$
    \EndIf
  \EndFor
  \State \Comment{\parbox[t]{.75\linewidth}{
  Use list tail-finding on the linked lists induced by \textproc{next\_edge} pointers.
  Get an array $\codevar{last\_marked}$ such that $\codevar{last\_marked}[i]$
  points to the last node in $z_i$'s linked list.}}
  \State $\codevar{last\_marked} = \Call{ListTailFind}{\codevar{elements}}$
  \State $\codevar{result} = \set{}$  \Comment{$k$-length array}
  \For{$i \in \set{1, \dots, k}$} \textbf{in parallel}
    \State $\codevar{result}[i] = \codevar{last\_marked}[i] \rightarrow \textproc{twin} \rightarrow \textproc{right}$
  \EndFor
  \State \Return $\codevar{result}$
  \EndProcedure
\end{algorithmic}
\end{algorithm}

\begin{algorithm}
  \caption{Remove a batch of edges from Euler tour tree.}
  \label{alg:et-batch-cut}
\small
\begin{algorithmic}[1]
  \Procedure{BatchCut}{$\set{\set{u_1, v_1}, \set{u_2, v_2}, \dots, \set{u_k, v_k}}$}
     \State \Comment Input edges must be in graph and must have no duplicates.
  \label{line:batch-cut-get-edges-start}
  \State $\codevar{directed\_edges} = \set{}$ \Comment $2k$-length array
  \For{$i \in \set{1, \ldots, k}$} \textbf{in parallel}
    \State $\codevar{directed\_edges}[2i - 1] = \codevar{edges}[(u_i, v_i)]$
    \State $\codevar{directed\_edges}[2i] = \codevar{edges}[(v_i, u_i)]$
  \EndFor
  \label{line:batch-cut-get-edges-end}
  \For{$i \in \set{1, \ldots, k}$} \textbf{in parallel}
    \State $\codevar{edges} \rightarrow \Call{RemoveFromDict}{(u_i, v_i)}$
    \State $\codevar{edges} \rightarrow \Call{RemoveFromDict}{(v_i, u_i)}$
  \EndFor
  \State $\codevar{join\_lefts} = \set{}$ \Comment $2k$-length array
  \For{$i \in \set{1, \ldots, 2k}$} \textbf{in parallel}
    \State $\codevar{join\_lefts}[i] = \codevar{directed\_edges}[i] \rightarrow \textproc{left}$
    \State $\codevar{directed\_edges}[i] \rightarrow \textproc{mark} = \codevar{true}$
    \label{line:batch-cut-mark-all}
  \EndFor
  \State $\codevar{join\_rights} = \Call{GetNextUnmarked}{\codevar{directed\_edges}}$
  \label{line:batch-cut-get-next-unmarked}
  \State \Comment Cut edges out of tour
  \label{line:batch-cut-split-start}
  \For{$i \in \set{1, \ldots, 2k}$} \textbf{in parallel}
    \State \Call{Split}{$\codevar{directed\_edges}[i]$}
    \State $\codevar{pred} = \codevar{directed\_edges}[i] \rightarrow \textproc{Left}$
    \If{$\codevar{pred} \not= \codevar{null}$}
      \State $\Call{Split}{\codevar{pred}}$
    \EndIf
  \EndFor
  \label{line:batch-cut-split-end}
  \State \Comment{Join tours back together}
  \label{line:batch-cut-join-start}
  \For{$i \in \set{1, \ldots, 2k}$} \textbf{in parallel}
    \If{\textbf{not} $\codevar{join\_lefts}[i]\rightarrow \textproc{mark}$}
      \State \Call{Join}{$\codevar{join\_lefts}[i], \codevar{join\_rights}[i]$}
    \EndIf
  \EndFor
  \label{line:batch-cut-join-end}
  \For{$i \in \set{1, \ldots, 2k}$} \textbf{in parallel}
    \State \Call{DeleteNode}{$\codevar{directed\_edges}[i]$}
  \EndFor
  \EndProcedure
\end{algorithmic}
\end{algorithm}

\myparagraph{Batch Cut}
Algorithm~\ref{alg:et-batch-cut} describes how to remove a batch of edges. Our algorithm assumes that each edge exists in
the forest and that there are no duplicates.

Cutting a single edge is simple. If we cut an edge $\set{u, v}$, we split before
and after $(u,v)$ and $(v, u)$ in the tour and join their neighbors together
appropriately. However, as with batch linking, the task gets more
difficult if we want to cut many edges out of a single node, because those
neighbors that we want to join together may themselves be split off.

\begin{figure}
    \centering
    \includegraphics[keepaspectratio=true, width=0.4\textwidth]{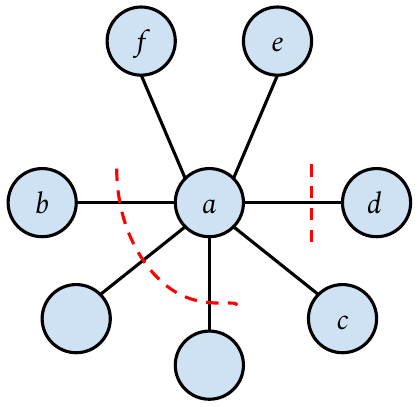}
    \caption[Tree undergoing batch cut] {Batch cutting four
    edges. If we take an Euler tour counter-clockwise around this graph, this
    batch cuts may require us to join $(f,a)$ to $(a, c)$ in the
    tour.}
    \label{fig:et-cut-example-graph}
\end{figure}

As an example, consider the graph in Figure~\ref{fig:et-cut-example-graph} in which
we remove four edges. If our tour on the graph goes counter-clockwise around the
diagram, then we may need to join $(c, a)$ to $(a, e)$ in the tour as a result
of cutting $\set{a,d}$ and join $(f,a)$ all the way around to $(a, c)$ as a
result of the three contiguous cuts. How do we identify that we need to join
$(f, a)$ to $(a, c)$? We could mark edges that are going to be cut, then start
from $(f,a)$ and walk along ``adjacent'' edges incident to $a$ using the
\textproc{twin} pointers until we reach an edge
that will not be cut. Then we would know to join $(f,a)$ to that edge. However,
the search for an unmarked edge will have poor depth if lots of
edges will be cut.

To achieve low depth in this step, we use list tail-finding. Consider
the linked lists induced by having each edge point at its adjacent
edge if it is marked. Note that each linked list must terminate
because traversing adjacent edges will eventually reach a loop edge of
the form $(v,v)$, which will certainly be unmarked. Then running list
tail-finding on these linked lists finds for every edge the next
unmarked edge as desired.

In Algorithm~\ref{alg:et-batch-cut}, we first fetch all the skip list nodes
corresponding to the edges in lines
\ref{line:batch-cut-get-edges-start}-\ref{line:batch-cut-get-edges-end}. Then we
invoke Algorithm~\ref{alg:list-rank} on
line~\ref{line:batch-cut-get-next-unmarked}, which performs the list
tail-finding described above. We cut out all the input edges on
lines~\ref{line:batch-cut-split-start}-\ref{line:batch-cut-split-end} and
rejoin all the tours together on
lines~\ref{line:batch-cut-join-start}-\ref{line:batch-cut-join-end}.
In total these steps take $O(k \log (1 + n/k))$ expected work and
$O(\log n)$ depth with high probability.

\myparagraph{Augmentation}
We build our augmented Euler tour trees over the concurrent augmented
skip lists from Subsection~\ref{subsec:augment-skip-lists} and achieve
the same efficiency bounds. Recall that we have an associative and
commutative function $f: D^2 \to D$ and assign values from $D$ to
vertices and edges of the forest. The goal is to compute $f$ over
subtrees of the represented forest.

Say we want to compute $f$ over a vertex $v$'s subtree relative to
$v$'s parent in the tree, $p$. Then if we look up the skip list
elements corresponding to $(p,v)$ and $(v,p)$ in $\codevar{edges}$,
the value of $f$ over $v$'s subtree is the result of applying $f$ on
the subsequence between $(p,v)$ and $(v,p)$. This may be done by
calling \textproc{BatchQueryValue} with the elements corresponding to
$(p, v)$ and $(v, p)$ in the underlying augmented skip list. The complexity for
$k$ such queries is $O(k \log n)$ expected work and $O(\log n)$ depth with high
probability.

\subsection{Implementation}
We provide details about our implementation of Euler tour trees in
Appendix~\ref{sec:app-ett-implementation}.

\section{Experiments}\label{sec:experiments}

We run our experiments on a 72-core Dell PowerEdge R930 (with two-way
hyper-threading) with 4 $\times$ 2.4GHz Intel 18-core E7-8867 v4 Xeon processors
(with a 4800MHz bus and 45MB L3 cache) and 1TB of main memory. Our programs use
Cilk Plus to express parallelism and are compiled with the g++ compiler (version
5.5.0).  When running in parallel, we
use the command \code{numactl -i all} to evenly distribute the allocated memory
among the processors. On our figures, a thread count of 72(h) denotes
using all $72$ cores with hyper-threading, i.e. using 144 threads.

\subsection{Unaugmented Skip Lists}\label{subsec:skip-list-experiment}

We evaluate the performance of our skip lists (with the probability of a node
having a direct parent set to $p = 1/2$) by comparing them against other
sequence data structures. In particular, we compare against sequential skip
lists, which are the same as our skip lists except that they do not use CAS to
set pointers. In addition, for an element of height $h$, they allocate an array
of exactly length $h$ for holding pointers rather than an array of length $O(h)$
as our parallel skip list implementation does (see Appendix~\ref{sec:app-skip-list-implementation}). We also
implemented splay trees~\cite{sleator1985self} and
treaps~\cite{aragon1989randomized}.

So that we can compare against another parallel data structure, we implement
parallel batch join and batch split operations on treaps. To batch join, we
first ignore a constant fraction of the joins. If we imagine each join from treap
$T$ to treap $S$ as a pointer from $T$ to $S$, we get lists on
the treaps. No list can be very long because of the ignored joins. We get
parallelism by processing each list independently. If we store extra information
on the treap nodes, we can walk along a list and perform its joins sequentially.
Then we recursively process the previously ignored joins.  For batch split, we
semisort the splits keyed on the root of the treap to be split.  This lets us
find all splits that act on a particular treap. We process each treap
independently. When performing multiple splits on a treap, we get parallelism by
divide and conquer---we perform a random split and recursively split
the resulting two treaps in parallel. The randomized efficiency bounds are $O(k
\log n)$ work for batch join, $O(k \log n \log k)$ work for batch split, and
$O(\log n \log k)$ depth for both. In the future, we would like to further
compare our skip lists against other parallel data structures, such as the
$(a,b)$-trees of Akhremtsev and Sanders~\cite{akhremtsev2016fast}.

For an experiment, we take $n = 10^8$ elements and fix a batch size
$k$. We set up a trial by joining all the elements in a chain, and then we time
how long it takes to split and rejoin the sequence at $k$ pseudorandomly sampled
locations. We report the median time over three trials. As an artifact of this
setup, the splay tree has an advantage on joining small batches after splitting due to how
splay trees exploit locality.

\begin{table} \centering
  \tiny
\begin{tabular}{l|l|l|l|l|l|l|l|l|l|l|l}
  \toprule
    \multirow{2}{*}{Data structure} & \multirow{2}{*}{\makecell[cl]{Batch\\size $k$}} & \multirow{2}{*}{Operation} & \multicolumn{8}{c}{Number of threads} \\
                                        &                         &       & 1     & 2     & 4     & 8      & 16     & 32     & 64      & 72      & 72(h) \\
  \hline
  \multirow{4}{*}{Concurrent skip list} & \multirow{2}{*}{$10^4$} & join  & .0301 & .0165 & .0101 & .00632 & .00298 & .00166 & .000937 & .000839 & .000530 \\
                                        &                         & split & .0331 & .0173 & .0113 & .00579 & .00298 & .00154 & .000865 & .000775 & .000542 \\
  \cline{2-12}                          & \multirow{2}{*}{$10^7$} & join  & 12.6  & 6.43  & 4.07  & 2.05   & 1.03   & .528   & .279    & .267    & .156 \\
                                        &                         & split & 10.4  & 5.34  & 3.34  & 1.86   & .869   & .426   & .228    & .214    & .122 \\
  \bottomrule
\end{tabular}
  \caption[Phase-concurrent skip list running times against number of
  threads]{Running time (in seconds) of our concurrent skip lists with $n = 10^8$.}
  \label{tab:pc-sl-threads}
\end{table}
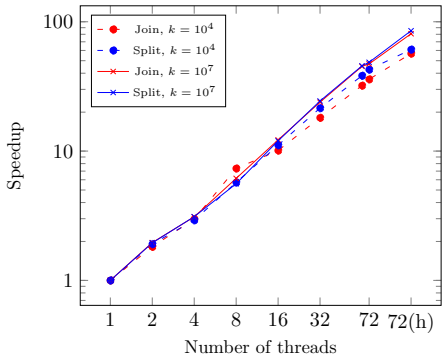
\begin{figure} \centering
    \begin{tikzpicture}[scale=0.7]
    \begin{axis}[
    legend style={font=\tiny},
    xmode=log,
    ymode=log,
    xlabel={\small Number of threads},
    ylabel={\small Speedup},
    xtick={1,2,4,8,16,32,72},
    extra x ticks = {144},
    extra x tick labels = {72(h)},
    log ticks with fixed point,
    legend pos=north west,
    ]
      \addplot[mark=*,color=red, loosely dashed] coordinates {
        (1,   1)
        (2,   1.82)
        (4,   2.98)
        (8,   7.34)
        (16,  10.1)
        (32,  18.1)
        (64,  32.1)
        (72,  35.9)
        (144, 56.8)
      };
      \addlegendentry{Join, $k=10^4$}

      \addplot[mark=*,color=blue, loosely dashed] coordinates {
        (1,   1)
        (2,   1.91)
        (4,   2.92)
        (8,   5.72)
        (16,  11.1)
        (32,  21.5)
        (64,  38.3)
        (72,  42.7)
        (144, 61.0)
      };
      \addlegendentry{Split, $k=10^4$}

      \addplot[mark=x,color=red] coordinates {
        (1,   1)
        (2,   1.96)
        (4,   3.10)
        (8,   6.15)
        (16,  12.2)
        (32,  23.9)
        (64,  45.2)
        (72,  47.2)
        (144, 80.8)
      };
      \addlegendentry{Join, $k=10^7$}

      \addplot[mark=x,color=blue] coordinates {
        (1,   1)
        (2,   1.95)
        (4,   3.11)
        (8,   5.59)
        (16,  12.0)
        (32,  24.4)
        (64,  45.6)
        (72,  48.6)
        (144, 85.2)
      };
      \addlegendentry{Split, $k=10^7$}
    \end{axis}
  \end{tikzpicture}
  \caption[Phase-concurrent skip list speedup]{Speedup of our concurrent skip
  lists with $n = 10^8$.}
  \label{fig:pc-sl-threads}
\end{figure}%

\begin{table} \centering
  \scriptsize
\begin{tabular}{l|l|l|l|l|l|l|l|l}
  \toprule
    \multirow{2}{*}{Data structure} & \multirow{2}{*}{Operation} & \multicolumn{7}{c}{Batch size $k$} \\
                                                 &       & $10^2$   & $10^3$  & $10^4$  & $10^5$ & $10^6$ & $10^7$ & $10^8 - 1$ \\
  \hline
  \multirow{2}{*}{Concurrent skip list (72(h))}  & join  & .0000629 & .000122 & .000595 & .00461 & .0347  & .161   & .684  \\
                                                 & split & .0000629 & .000132 & .000660 & .00363 & .0254  & .131   & .559  \\
  \hline
  \multirow{2}{*}{Concurrent skip list (1)}      & join  & .000300  & .00260  & .0295   & .376   & 2.44   & 12.8   & 54.7  \\
                                                 & split & .000383  & .00355  & .0325   & .281   & 1.90   & 10.6   & 47.6  \\
  \hline
  \multirow{2}{*}{Sequential skip list}          & join  & .000324  & .00265  & .0275   & .362   & 2.26   & 11.7   & 44.9  \\
                                                 & split & .000396  & .00359  & .0319   & .272   & 1.80   & 9.87   & 45.0  \\
  \hline
  \multirow{2}{*}{Parallel treap (72(h))}        & join  & .000227  & .000758 & .00141  & .00475 & .0250  & .112   & .447  \\
                                                 & split & .000335  & .00186  & .00521  & .0277  & .265   & 2.54   & 22.8  \\
  \hline
  \multirow{2}{*}{Sequential treap}              & join  & .0000989 & .00104  & .00712  & .183   & 1.23   & 6.86   & 25.2  \\
                                                 & split & .000231  & .00213  & .0189   & .168   & 1.30   & 7.69   & 33.5  \\
  \hline
  \multirow{2}{*}{Splay tree}                    & join  & .0000689 & .000688 & .00575  & .106   & 1.09   & 7.79   & 32.1  \\
                                                 & split & .000329  & .00284  & .0255   & .215   & 1.63   & 9.21   & 36.3  \\
  \bottomrule
\end{tabular}
  \caption[Sequence data structure running times with varying batch size]{Running
  time (in seconds) of sequence data structures with $n=10^8$ and varying batch size.}
  \label{tab:pc-sl-batch-size}
\end{table}

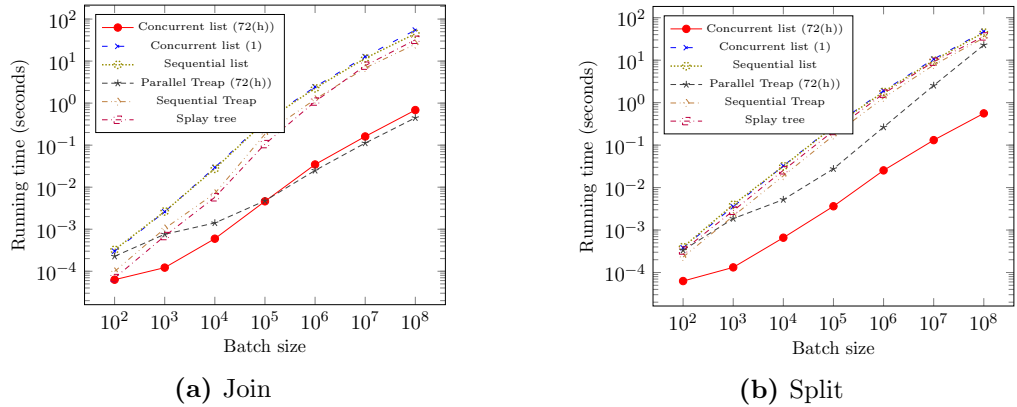
\begin{figure} \centering
  \begin{subfigure}{0.45 \textwidth} \centering
    \begin{tikzpicture}[scale=0.7]
    \begin{axis}[
    legend style={font=\tiny},
    xmode=log,
    ymode=log,
    xlabel={\small Batch size},
    ylabel={\small Running time (seconds)},
    xtick={1e2,1e3,1e4,1e5,1e6,1e7,1e8},
    xminorticks=false,
    legend pos=north west,
    ]
      \addplot[mark=*,color=red] coordinates {
        (1e2,   .0000629)
        (1e3,   .000122)
        (1e4,   .000595)
        (1e5,   .00461)
        (1e6,   .0347)
        (1e7,   .161)
        (1e8-1, .684)
      };
      \addlegendentry{Concurrent list (72(h))}

      \addplot[mark=x,color=blue, loosely dashed] coordinates {
        (1e2,   .0003)
        (1e3,   .0026)
        (1e4,   .0295)
        (1e5,   .376)
        (1e6,   2.44)
        (1e7,   12.8)
        (1e8-1, 54.7)
      };
      \addlegendentry{Concurrent list (1)}

      \addplot[mark=o,color=olive, thick, densely dotted] coordinates {
        (1e2,   .000324)
        (1e3,   .00265)
        (1e4,   .0275)
        (1e5,   .362)
        (1e6,   2.26)
        (1e7,   11.7)
        (1e8-1, 44.9)
      };
      \addlegendentry{Sequential list}

      \addplot[mark=star,color=darkgray, densely dashed] coordinates {
        (1e2,   .000227)
        (1e3,   .000758)
        (1e4,   .00141)
        (1e5,   .00475)
        (1e6,   0.025)
        (1e7,   0.112)
        (1e8-1, 0.447)
      };
      \addlegendentry{Parallel Treap (72(h))}

      \addplot[mark=diamond,color=brown, dashdotdotted] coordinates {
        (1e2,   .0000989)
        (1e3,   .00104)
        (1e4,   .00712)
        (1e5,   .183)
        (1e6,   1.23)
        (1e7,   6.86)
        (1e8-1, 25.2)
      };
      \addlegendentry{Sequential Treap}

      \addplot[mark=square,color=purple, dashdotted] coordinates {
        (1e2,   .0000689)
        (1e3,   .000688)
        (1e4,   .00575)
        (1e5,   .106)
        (1e6,   1.09)
        (1e7,   7.79)
        (1e8-1, 32.1)
      };
      \addlegendentry{Splay tree}
    \end{axis}
  \end{tikzpicture}
  \caption{Join}
  \end{subfigure}
  \begin{subfigure}{0.45 \textwidth} \centering
    \begin{tikzpicture}[scale=0.7]
    \begin{axis}[
    legend style={font=\tiny},
    xmode=log,
    ymode=log,
    xlabel={\small Batch size},
    ylabel={\small Running time (seconds)},
    xtick={1e2,1e3,1e4,1e5,1e6,1e7,1e8},
    xminorticks=false,
    legend pos=north west,
    ]
      \addplot[mark=*,color=red] coordinates {
        (1e2,   .0000629)
        (1e3,   .000132)
        (1e4,   .00066)
        (1e5,   .00363)
        (1e6,   .0254)
        (1e7,   .131)
        (1e8-1, .559)
      };
      \addlegendentry{Concurrent list (72(h))}

      \addplot[mark=x,color=blue, loosely dashed] coordinates {
        (1e2,   .000383)
        (1e3,   .00355)
        (1e4,   .0325)
        (1e5,   .281)
        (1e6,   1.9)
        (1e7,   10.6)
        (1e8-1, 47.6)
      };
      \addlegendentry{Concurrent list (1)}

      \addplot[mark=o,color=olive, thick, densely dotted] coordinates {
        (1e2,   .000396)
        (1e3,   .00395)
        (1e4,   .0319)
        (1e5,   .272)
        (1e6,   1.8)
        (1e7,   9.87)
        (1e8-1, 45)
      };
      \addlegendentry{Sequential list}

      \addplot[mark=star,color=darkgray, densely dashed] coordinates {
        (1e2,   .000335)
        (1e3,   .00186)
        (1e4,   .00521)
        (1e5,   .0277)
        (1e6,   0.265)
        (1e7,   2.54)
        (1e8-1, 22.8)
      };
      \addlegendentry{Parallel Treap (72(h))}

      \addplot[mark=diamond,color=brown, dashdotdotted] coordinates {
        (1e2,   .000231)
        (1e3,   .00213)
        (1e4,   .0189)
        (1e5,   .168)
        (1e6,   1.3)
        (1e7,   7.69)
        (1e8-1, 33.5)
      };
      \addlegendentry{Sequential Treap}

      \addplot[mark=square,color=purple, dashdotted] coordinates {
        (1e2,   .000329)
        (1e3,   .00284)
        (1e4,   .0255)
        (1e5,   .215)
        (1e6,   1.63)
        (1e7,   9.21)
        (1e8-1, 36.3)
      };
      \addlegendentry{Splay tree}
    \end{axis}
  \end{tikzpicture}
  \caption{Split}
  \end{subfigure}
  \caption[Sequence data structure running times with varying batch
  size]{Running time of sequence data structures operations with varying batch size.}
  \label{fig:pc-sl-batch-size}
\end{figure}%

Table~\ref{tab:pc-sl-threads} and Figure~\ref{fig:pc-sl-threads}
illustrate that our skip list implementation
running on $72$ cores with hyper-threading demonstrates over $80\times$ speedup
relative to the implementation running on a single thread for $k=10^7$ and over
$55\times$ speedup for $k=10^4$. We compare our skip list to our other sequence
data structures in Table~\ref{tab:pc-sl-batch-size} and
Figure~\ref{fig:pc-sl-batch-size} . Our implementation of parallel batch join on
treaps is $1.4\times$ faster than our batch join on skip lists on the largest
batch sizes, but, as seen in Figure~\ref{fig:pc-sl-batch-size}, the parallel
batch split on treaps is much slower due to lots of overhead work. Moreover,
through parallelism, our data structure is significantly faster than all the
sequential algorithms at all batch sizes. When used sequentially, our data
structure behaves similarly to a traditional sequential skip list, suggesting
that using CAS does not significantly degrade the performance of a skip list.

\subsection{Augmented Skip Lists}

We compare the performance of our batch-parallel augmented skip lists against
a sequential augmented skip list. Besides not using CAS, the sequential skip
list updates augmented values after every join and split. This achieves only an
$O(k \log n)$ work bound for $k$ operations. Our experiment is the same as in
Subsection~\ref{subsec:skip-list-experiment}.

\begin{table} \centering
  \scriptsize
\begin{tabular}{l|l|l|l|l|l|l|l|l|l|l|l}
  \toprule
    \multirow{2}{*}{Data structure} & \multirow{2}{*}{\makecell[cl]{Batch\\size $k$}} & \multirow{2}{*}{Operation} &\multicolumn{8}{c}{Number of threads} \\
                                       &                         &       & 1     & 2     & 4     & 8     & 16     & 32     & 64     & 72     & 72(h) \\
  \hline
  \multirow{4}{*}{Parallel skip list}  & \multirow{2}{*}{$10^4$} & join  & .0908 & .0457 & .0251 & .0164 & .00871 & .00487 & .00292 & .00282 & .00227 \\
                                       &                         & split & .0546 & .0283 & .0195 & .0134 & .00561 & .00306 & .00177 & .00164 & .00113 \\
  \cline{2-12}                         & \multirow{2}{*}{$10^7$} & join  & 24.9  & 12.7  & 9.39  & 4.41  & 2.10   & 1.09   & .614   & .61    & .374 \\
                                       &                         & split & 19.7  & 10.1  & 7.49  & 3.48  & 1.61   & .810   & .422   & .398   & .253 \\
  \bottomrule
\end{tabular}
  \caption[Batch-parallel augmented skip list running times against number of
  threads]{Running time (in seconds) of our parallel augmented skip lists with $n = 10^8$.}
  \label{tab:aug-sl-threads}
\end{table}

\begin{figure} \centering
    \begin{tikzpicture}[scale=0.7]
    \begin{axis}[
    legend style={font=\tiny},
    xmode=log,
    ymode=log,
    xlabel={\small Number of threads},
    ylabel={\small Speedup},
    xtick={1,2,4,8,16,32,72},
    extra x ticks = {144},
    extra x tick labels = {72(h)},
    log ticks with fixed point,
    legend pos=north west,
    ]
      \addplot[mark=*,color=red, loosely dashed] coordinates {
        (1,   1)
        (2,   1.99)
        (4,   3.62)
        (8,   5.53)
        (16,  10.4)
        (32,  18.6)
        (64,  31.1)
        (72,  32.2)
        (144, 40.0)
      };
      \addlegendentry{Join, $k=10^4$}

      \addplot[mark=*,color=blue, loosely dashed] coordinates {
        (1,   1)
        (2,   1.93)
        (4,   2.80)
        (8,   4.07)
        (16,  9.73)
        (32,  17.84)
        (64,  30.8)
        (72,  33.3)
        (144, 48.3)
      };
      \addlegendentry{Split, $k=10^4$}

      \addplot[mark=x,color=red] coordinates {
        (1,   1)
        (2,   1.96)
        (4,   2.65)
        (8,   5.65)
        (16,  11.9)
        (32,  22.8)
        (64,  40.6)
        (72,  40.8)
        (144, 66.6)
      };
      \addlegendentry{Join, $k=10^7$}

      \addplot[mark=x,color=blue] coordinates {
        (1,   1)
        (2,   1.95)
        (4,   2.63)
        (8,   5.66)
        (16,  12.2)
        (32,  24.3)
        (64,  46.7)
        (72,  49.5)
        (144, 77.9)
      };
      \addlegendentry{Split, $k=10^7$}
    \end{axis}
  \end{tikzpicture}
  \caption[Batch-parallel augmented skip list and speedup]{Speedup of our
  parallel augmented skip lists with $n = 10^8$.}
  \label{fig:aug-sl-threads}
\end{figure}
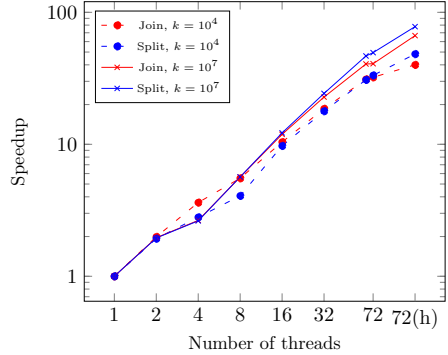

\begin{table} \centering
  \scriptsize
\begin{tabular}{l|l|l|l|l|l|l|l|l}
  \toprule
    \multirow{2}{*}{Data structure} & \multirow{2}{*}{Operation} & \multicolumn{7}{c}{Batch size $k$} \\
                                                &       & $10^2$  & $10^3$  & $10^4$ & $10^5$ & $10^6$ & $10^7$ & $10^8 - 1$ \\
  \hline
  \multirow{2}{*}{Parallel skip list (72(h))}   & join  & .000301 & .000559 & .00205 & .0131  & .0825  & .357   & 1.35  \\
                                                & split & .000152 & .000312 & .00117 & .00727 & .0450  & .229   & 1.03 \\
  \hline
  \multirow{2}{*}{Parallel skip list (1)}       & join  & .000823 & .00697  & .0893  & 1.04   & 5.73   & 25.5   & 102   \\
                                                & split & .00079  & .00625  & .0514  & .557   & 3.74   & 20.2   & 83.7  \\
  \hline
  \multirow{2}{*}{Sequential skip list}         & join  & .000716 & .00575  & .0764  & .724   & 4.79   & 27.5   & 131   \\
                                                & split & .000712 & .00647  & .0583  & .489   & 3.35   & 19.6   & 93.3  \\
  \bottomrule
\end{tabular}
  \caption[Augmented skip list running times with varying batch size]{Running time (in seconds) of
  augmented skip lists with $n = 10^8$ on random batches of varying size.}
  \label{tab:aug-sl-batch-size-random}
\end{table}

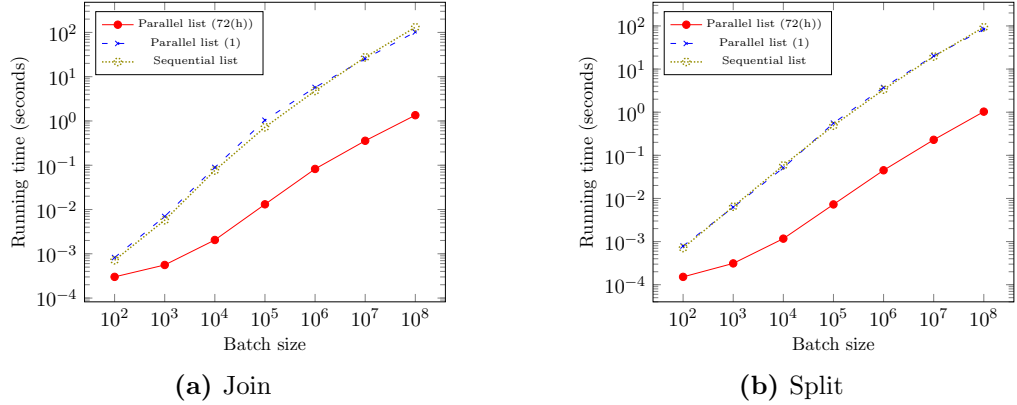
\begin{figure} \centering
  \begin{subfigure}{0.45 \textwidth} \centering
    \begin{tikzpicture}[scale=0.7]
    \begin{axis}[
    legend style={font=\tiny},
    xmode=log,
    ymode=log,
    xlabel={\small Batch size},
    ylabel={\small Running time (seconds)},
    xtick={1e2,1e3,1e4,1e5,1e6,1e7,1e8},
    xminorticks=false,
    legend pos=north west,
    ]
      \addplot[mark=*,color=red] coordinates {
        (1e2,   .000301)
        (1e3,   .000559)
        (1e4,   .00205)
        (1e5,   .0131)
        (1e6,   .0825)
        (1e7,   .357)
        (1e8-1, 1.35)
      };
      \addlegendentry{Parallel list (72(h))}

      \addplot[mark=x,color=blue, loosely dashed] coordinates {
        (1e2,   .000823)
        (1e3,   .00697)
        (1e4,   .0893)
        (1e5,   1.04)
        (1e6,   5.73)
        (1e7,   25.5)
        (1e8-1, 102)
      };
      \addlegendentry{Parallel list (1)}

      \addplot[mark=o,color=olive, thick, densely dotted] coordinates {
        (1e2,   .000716)
        (1e3,   .00575)
        (1e4,   .0764)
        (1e5,   .724)
        (1e6,   4.79)
        (1e7,   27.5)
        (1e8-1, 131)
      };
      \addlegendentry{Sequential list}
    \end{axis}
  \end{tikzpicture}
  \caption{Join}
  \end{subfigure}
  \begin{subfigure}{0.45 \textwidth} \centering
    \begin{tikzpicture}[scale=0.7]
    \begin{axis}[
    legend style={font=\tiny},
    xmode=log,
    ymode=log,
    xlabel={\small Batch size},
    ylabel={\small Running time (seconds)},
    xtick={1e2,1e3,1e4,1e5,1e6,1e7,1e8},
    xminorticks=false,
    legend pos=north west,
    ]
      \addplot[mark=*,color=red] coordinates {
        (1e2,   .000152)
        (1e3,   .000312)
        (1e4,   .00117)
        (1e5,   .00727)
        (1e6,   .045)
        (1e7,   .229)
        (1e8-1, 1.03)
      };
      \addlegendentry{Parallel list (72(h))}

      \addplot[mark=x,color=blue, loosely dashed] coordinates {
        (1e2,   .00079)
        (1e3,   .00625)
        (1e4,   .0514)
        (1e5,   .557)
        (1e6,   3.74)
        (1e7,   20.2)
        (1e8-1, 83.7)
      };
      \addlegendentry{Parallel list (1)}

      \addplot[mark=o,color=olive, thick, densely dotted] coordinates {
        (1e2,   .000712)
        (1e3,   .00647)
        (1e4,   .0583)
        (1e5,   .489)
        (1e6,   3.35)
        (1e7,   19.6)
        (1e8-1, 93.3)
      };
      \addlegendentry{Sequential list}
    \end{axis}
  \end{tikzpicture}
  \caption{Split}
  \end{subfigure}
  \caption[Augmented skip list running times with varying batch size]{Running
  time of augmented sequence data structure operations with $n = 10^8$ on random
  batches of varying size.}
  \label{fig:aug-sl-batch-size-random}
\end{figure}

Table~\ref{tab:aug-sl-threads} and Figure~\ref{fig:aug-sl-threads} show that
when running our augmented skip list with a random batch of size $k=10^7$ on
$72$ cores with hyper-threading, we see a speedup of $67\times$ for joins and
$78\times$ for splits. For $k=10^4$, we found a speedup of $33\times$ for joins
and $48\times$ for splits. The running times are a factor of two worse than the
times for the unaugmented skip list of
Subsection~\ref{subsec:skip-list-experiment}, which is expected due to the extra
passes through the skip list to update augmented values. Moreover, the
batch-parallel skip list hugely outperforms single-threaded skip lists on all
tested batch sizes, as seen in Table~\ref{tab:aug-sl-batch-size-random} and
Figure~\ref{fig:aug-sl-batch-size-random}.

\begin{table} \centering
  \scriptsize
\begin{tabular}{l|l|l|l|l|l|l|l|l}
  \toprule
    \multirow{2}{*}{Data structure} & \multirow{2}{*}{Operation} &\multicolumn{7}{c}{Batch size $k$} \\
                              &       & $10^2$   & $10^3$  & $10^4$  & $10^5$  & $10^6$ & $10^7$ & $10^8 - 1$ \\
  \hline
  Parallel skip list (72(h))  & split & .000111  & .000192 & .000272 & .000596 & .00281 & .0259  & .252 \\
  \hline
  Parallel skip list (1)      & split & .0000350 & .000148 & .00119  & .0115   & .119   & 1.20   & 11.9  \\
  \hline
  Sequential skip list        & split & .000296  & .00258  & .0219   & .210    & 2.11   & 21.5   & 205  \\
  \bottomrule
\end{tabular}
  \caption[Augmented skip list running times on an adversarial test case]{Running
  time (in seconds) of splitting augmented skip lists with $n= 10^8$ as batch size varies with
  splits taking single elements off the end of the list. This is a difficult
  test case for standard augmented skip lists.}
  \label{tab:aug-sl-batch-size-adverse}
\end{table}
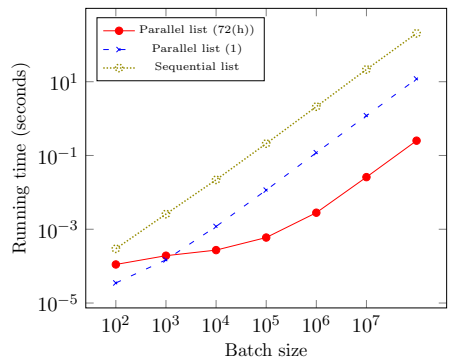
\begin{figure} \centering
  \begin{tikzpicture}[scale=0.7]
  \begin{axis}[
  legend style={font=\tiny},
  xmode=log,
  ymode=log,
  xlabel={\small Batch size},
  ylabel={\small Running time (seconds)},
  xtick={1e2,1e3,1e4,1e5,1e6,1e7},
  xminorticks=false,
  legend pos=north west,
  ]
    \addplot[mark=*,color=red] coordinates {
      (1e2,   .000111)
      (1e3,   .000192)
      (1e4,   .000272)
      (1e5,   .000596)
      (1e6,   .00281)
      (1e7,   .0259)
      (1e8-1, .252)
    };
    \addlegendentry{Parallel list (72(h))}

    \addplot[mark=x,color=blue, loosely dashed] coordinates {
      (1e2,   .000035)
      (1e3,   .000148)
      (1e4,   .00119)
      (1e5,   .0115)
      (1e6,   .119)
      (1e7,   1.2)
      (1e8-1, 11.9)
    };
    \addlegendentry{Parallel list (1)}

    \addplot[mark=o,color=olive, thick, densely dotted] coordinates {
      (1e2,   .000296)
      (1e3,   .00258)
      (1e4,   .0219)
      (1e5,   .21)
      (1e6,   2.11)
      (1e7,   21.5)
      (1e8-1, 205)
    };
    \addlegendentry{Sequential list}
  \end{axis}
  \end{tikzpicture}
  \caption[Augmented skip list running times on an adversarial test case] {
    Running time of splitting augmented skip lists with $n= 10^8$ as batch size varies with
  splits taking off single elements off the end of the list.}
  \label{fig:aug-sl-batch-size-adverse}
\end{figure}%

To more prominently display the work savings that batching provides, we try a
different test case in which a batch of $k$ splits, rather than being chosen at random,
consists of splitting at the last $k$ elements of the sequence in right-to-left order.
This is particularly bad for the sequential skip list because after every split,
it walks to the top of the skip list to update augmented values. In comparison,
when processing the splits as a batch, we update the augmented values in only
one pass.
Table~\ref{tab:aug-sl-batch-size-adverse} and
Figure~\ref{fig:aug-sl-batch-size-adverse} show that, as expected, even on a
single thread, our skip list is significantly faster than the standard
sequential one in this adversarial experiment.

\subsection{Euler Tour Trees}

We compare against sequential dynamic trees data structures. Using the
sequential skip list and splay trees from
Subsection~\ref{subsec:skip-list-experiment}, we build traditional Euler tour
trees. We also compare to ST-trees built on splay
trees~\cite{sleator1983data,sleator1985self}. They achieve $O(\log n)$ amortized
work links and cuts. Though conceptually more complicated than Euler tour trees,
ST-trees are a more streamlined data structure that do not require allocation
beyond initialization and do not require dictionary lookups. We wrote all of
these implementations. In future work, we would like to compare against parallel
data structures such as that of Acar et al.~\cite{acar2017brief}

Because one of the important uses of Euler tour trees is to answer connectivity
queries, we also compare with statically computing the connected components
of the graph. We use the work-efficient parallel
connectivity algorithm designed and implemented by Shun et
al.~\cite{shun2014simple}. We optimistically measure the execution
time of the implementation based only on the execution time of the
connectivity algorithm; we do not include the time taken to maintain
the graph itself, which is non-trivial because the adjacency array
graph representation used in their implementation does not support
edge insertion or deletion easily.

For our experiment, we fix a tree. We set up a trial by adding all the edges of
the tree in pseudorandom order to our data stucture. Then we time how long it
takes to cut and relink the forest at $k$ pseudorandomly sampled edges. We
report median times over three trials. Again, note that our
experimental setup may give the splay tree data structures an
advantage on linking small batches after cutting due to how splay
trees exploit locality. We experimented on three trees, all with $n = 10^7$
vertices: a path graph, a star graph, and a random recursive tree. To form a
random recursive tree over $n$ vertices, for each $1 < i \le n$, draw $j$
uniformly at random from $\set{1, 2, \ldots, i -1}$ and add the edge $\set{j,
i}$.

\begin{table} \centering
  \tiny
\begin{tabular}{l|l|l|l|l|l|l|l|l|l|l|l|l}
  \toprule
    \multirow{2}{*}{Data structure} &
    \multirow{2}{*}{Graph} &
    \multirow{2}{*}{\makecell[cl]{Batch\\ size $k$}} &
    \multirow{2}{*}{Operation} &
    \multicolumn{8}{c}{Number of threads} \\
                                  &                                                         &                         &      & 1     & 2     & 4      & 8      & 16     & 32     & 64     & 72     & 72(h) \\
  \hline
  \multirow{12}{*}{Parallel ETT}  & \multirow{4}{*}{\makecell[cl]{Path\\graph}}             & \multirow{2}{*}{$10^4$} & link & .175  & .109  & .0559  & .0172  & .00988 & .00595 & .00374 & .00345 & .0029\\
                                  &                                                         &                         & cut  & .185  & .129  & .0641  & .0270  & .0139  & .00743 & .00417 & .00378 & .00267 \\
  \cline{3-13}                    &                                                         & \multirow{2}{*}{$10^6$} & link & 7.25  & 5.10  & 2.53   & 1.10   & .548   & .279   & .143   & .131   & .0813\\
                                  &                                                         &                         & cut  & 8.74  & 6.48  & 3.22   & 1.36   & .672   & .348   & .177   & .159   & .0980 \\
  \cline{2-13}                    & \multirow{4}{*}{\makecell[cl]{Random\\recursive\\tree}} & \multirow{2}{*}{$10^4$} & link & .111  & .0579 & .0266  & .0162  & .00849 & .00524 & .00334 & .00316 & .00278\\
                                  &                                                         &                         & cut  & .149  & .0759 & .0358  & .0192  & .00962 & .00525 & .00301 & .00275 & .00199 \\
  \cline{3-13}                    &                                                         & \multirow{2}{*}{$10^6$} & link & 8.77  & 6.31  & 3.20   & 1.34   & .684   & .344   & .182   & .161   & .0972  \\
                                  &                                                         &                         & cut  & 7.87  & 5.87  & 2.96   & 1.26   & .628   & .323   & .171   & .147   & .0882   \\
  \cline{2-13}                    & \multirow{4}{*}{\makecell[cl]{Star\\graph}}             & \multirow{2}{*}{$10^4$} & link & .0154 & .0147 & .00693 & .00533 & .00344 & .00249 & .00203 & .00191 & .00198\\
                                  &                                                         &                         & cut  & .0170 & .0112 & .00733 & .00442 & .00247 & .00136 & .00102 & .00102 & .000922 \\
  \cline{3-13}                    &                                                         & \multirow{2}{*}{$10^6$} & link & 3.49  & 2.05  & 1.01   & .454   & .237   & .125   & .0681  & .0618  & .0408\\
                                  &                                                         &                         & cut  & 2.60  & 1.56  & .740   & .339   & .171   & .0883  & .0467  & .0419  & .0271 \\
  \bottomrule
\end{tabular}
  \caption[Parallel Euler tour tree running times against number of
  threads]{Running time (in seconds) of our batch-parallel Euler tour tree on various graphs with $n = 10^7$.}
  \label{tab:ett-threads}
\end{table}

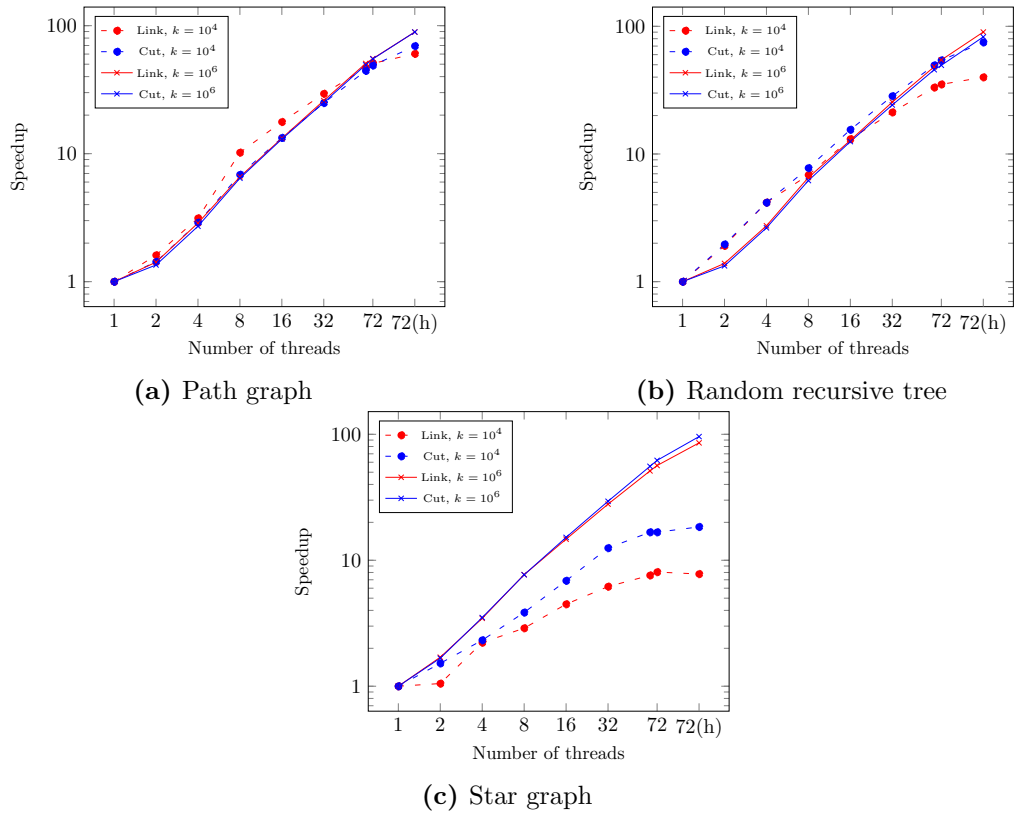
\begin{figure} \centering
  \begin{subfigure}{0.45 \textwidth} \centering
    \begin{tikzpicture}[scale=0.7]
    \begin{axis}[
    legend style={font=\tiny},
    xmode=log,
    ymode=log,
    xlabel={\small Number of threads},
    ylabel={\small Speedup},
    xtick={1,2,4,8,16,32,72},
    extra x ticks = {144},
    extra x tick labels = {72(h)},
    log ticks with fixed point,
    legend pos=north west,
    ]
      \addplot[mark=*,color=red, loosely dashed] coordinates {
        (1,   1)
        (2,   1.61)
        (4,   3.13)
        (8,   10.2)
        (16,  17.7)
        (32,  29.4)
        (64,  46.8)
        (72,  50.7)
        (144, 60.3)
      };
      \addlegendentry{Link, $k=10^4$}

      \addplot[mark=*,color=blue, loosely dashed] coordinates {
        (1,   1)
        (2,   1.43)
        (4,   2.89)
        (8,   6.85)
        (16,  13.3)
        (32,  24.9)
        (64,  44.4)
        (72,  48.9)
        (144, 69.3)
      };
      \addlegendentry{Cut, $k=10^4$}

      \addplot[mark=x,color=red] coordinates {
        (1,   1)
        (2,   1.42)
        (4,   2.87)
        (8,   6.59)
        (16,  13.2)
        (32,  26.0)
        (64,  50.7)
        (72,  55.3)
        (144, 89.2)
      };
      \addlegendentry{Link, $k=10^6$}

      \addplot[mark=x,color=blue] coordinates {
        (1,   1)
        (2,   1.35)
        (4,   2.71)
        (8,   6.43)
        (16,  13.0)
        (32,  25.1)
        (64,  49.4)
        (72,  55.0)
        (144, 89.2)
      };
      \addlegendentry{Cut, $k=10^6$}
    \end{axis}
  \end{tikzpicture}
  \caption{Path graph}
  \end{subfigure}
  \begin{subfigure}{0.45 \textwidth} \centering
    \begin{tikzpicture}[scale=0.7]
    \begin{axis}[
    legend style={font=\tiny},
    xmode=log,
    ymode=log,
    xlabel={\small Number of threads},
    ylabel={\small Speedup},
    xtick={1,2,4,8,16,32,72},
    extra x ticks = {144},
    extra x tick labels = {72(h)},
    log ticks with fixed point,
    legend pos=north west,
    ]
      \addplot[mark=*,color=red, loosely dashed] coordinates {
        (1,   1)
        (2,   1.91)
        (4,   4.17)
        (8,   6.85)
        (16,  13.1)
        (32,  21.2)
        (64,  33.2)
        (72,  35.1)
        (144, 39.9)
      };
      \addlegendentry{Link, $k=10^4$}

      \addplot[mark=*,color=blue, loosely dashed] coordinates {
        (1,   1)
        (2,   1.96)
        (4,   4.16)
        (8,   7.76)
        (16,  15.5)
        (32,  28.4)
        (64,  49.5)
        (72,  54.2)
        (144, 74.9)
      };
      \addlegendentry{Cut, $k=10^4$}

      \addplot[mark=x,color=red] coordinates {
        (1,   1)
        (2,   1.39)
        (4,   2.74)
        (8,   6.54)
        (16,  12.8)
        (32,  25.5)
        (64,  48.2)
        (72,  54.5)
        (144, 90.2)
      };
      \addlegendentry{Link, $k=10^6$}

      \addplot[mark=x,color=blue] coordinates {
        (1,   1)
        (2,   1.33)
        (4,   2.64)
        (8,   6.19)
        (16,  12.5)
        (32,  24.2)
        (64,  45.7)
        (72,  49.5)
        (144, 82.6)
      };
      \addlegendentry{Cut, $k=10^6$}
    \end{axis}
  \end{tikzpicture}
  \caption{Random recursive tree}
  \end{subfigure}
  \begin{subfigure}{0.45 \textwidth} \centering
    \begin{tikzpicture}[scale=0.7]
    \begin{axis}[
    legend style={font=\tiny},
    xmode=log,
    ymode=log,
    xlabel={\small Number of threads},
    ylabel={\small Speedup},
    xtick={1,2,4,8,16,32,72},
    extra x ticks = {144},
    extra x tick labels = {72(h)},
    log ticks with fixed point,
    legend pos=north west,
    ]
      \addplot[mark=*,color=red, loosely dashed] coordinates {
        (1,   1)
        (2,   1.05)
        (4,   2.22)
        (8,   2.89)
        (16,  4.48)
        (32,  6.18)
        (64,  7.59)
        (72,  8.06)
        (144, 7.78)
      };
      \addlegendentry{Link, $k=10^4$}

      \addplot[mark=*,color=blue, loosely dashed] coordinates {
        (1,   1)
        (2,   1.52)
        (4,   2.32)
        (8,   3.85)
        (16,  6.88)
        (32,  12.5)
        (64,  16.7)
        (72,  16.7)
        (144, 18.4)
      };
      \addlegendentry{Cut, $k=10^4$}

      \addplot[mark=x,color=red] coordinates {
        (1,   1)
        (2,   1.70)
        (4,   3.46)
        (8,   7.69)
        (16,  14.7)
        (32,  27.9)
        (64,  51.2)
        (72,  56.5)
        (144, 85.5)
      };
      \addlegendentry{Link, $k=10^6$}

      \addplot[mark=x,color=blue] coordinates {
        (1,   1)
        (2,   1.67)
        (4,   3.51)
        (8,   7.67)
        (16,  15.2)
        (32,  29.4)
        (64,  55.7)
        (72,  62.1)
        (144, 95.9)
      };
      \addlegendentry{Cut, $k=10^6$}
    \end{axis}
  \end{tikzpicture}
  \caption{Star graph}
  \end{subfigure}
  \caption[Batch-parallel Euler tour tree speedup]{Speedup of our parallel Euler
  tour trees on a various forests of size $n = 10^7$.}
  \label{fig:ett-threads}
\end{figure}%

\begin{table} \centering
  \scriptsize
\begin{tabular}{l|l|l}
  \toprule
    Data structure                                & Graph                 & \\
  \hline
     \multirow{3}{*}{Static connectivity (72(h))} & Path graph            & .576 \\
     \cline{2-3}                                  & Random recursive tree & .174 \\
     \cline{2-3}                                  & Star graph            & .113 \\
  \bottomrule
\end{tabular}
  \caption[Static connectivity on various graphs with $n=10^7$]{Running time (in seconds) of static connectivity on various graphs with $n = 10^7$}
  \label{tab:static-conn}
\end{table}

\begin{table} \centering
  \scriptsize
\begin{tabular}{l|l|l|l|l|l|l|l|l}
  \toprule
    \multirow{2}{*}{Graph} &
    \multirow{2}{*}{Data structure} &
    \multirow{2}{*}{Operation} &
    \multicolumn{6}{c}{Batch size $k$} \\
                                                           &                                       &      & $10^2$   & $10^3$  & $10^4$ & $10^5$ & $10^6$ & $10^7 - 1$\\
  \hline
  \multirow{10}{*}{\makecell[cl]{Path\\graph}}             & \multirow{2}{*}{Parallel ETT (72(h))} & link & .000276  & .000642  & .00290 & .0165  & .0813  & .305 \\
                                                           &                                       & cut  & .000297  & .000714  & .00267 & .0169  & .0980   & .408 \\
   \cline{2-9}                                             & \multirow{2}{*}{Parallel ETT (1)}     & link & .000950  & .00784  & .175   & 1.29   & 7.25   & 29.2\\
                                                           &                                       & cut  & .00223   & .0187   & .185   & 1.40   & 8.74   & 37.6 \\
   \cline{2-9}                                             & \multirow{2}{*}{Seq.\ skip list ETT}  & link & .000767  & .00617  & .131   & 1.04   & 6.77   & 33.4 \\
                                                           &                                       & cut  & .00148   & .0144   & .129   & 1.03   & 6.73   & 35.0 \\
   \cline{2-9}                                             & \multirow{2}{*}{Splay ETT}            & link & .000408    & .00345   & .0631  & .605   & 4.82   & 27.3 \\
                                                           &                                       & cut  & .00109     & .0103   & .0922  & .484   & 5.51   & 31.7 \\
   \cline{2-9}                                             & \multirow{2}{*}{ST-tree}              & link & .0000420 & .000418 & .00499 & .0818  & 1.04   & 5.09 \\
                                                           &                                       & cut  & .000562     & .00471   & .0405  & .310   & 1.88   & 7.27 \\
  \hline
  \multirow{10}{*}{\makecell[cl]{Random\\recursive\\tree}} & \multirow{2}{*}{Parallel ETT (72(h))} & link & .000208  & .000584  & .00279 & .0153  & .0972   & .430\\
                                                           &                                       & cut  & .000237  & .000599  & .00195 & .013  & .0882   & .420 \\
   \cline{2-9}                                             & \multirow{2}{*}{Parallel ETT (1)}     & link & .000674  & .00560  & .144   & 1.25   & 8.77   & 42.3 \\
                                                           &                                       & cut  & .00113   & .0100   & .137   & 1.12   & 7.82   & 37.9 \\
   \cline{2-9}                                             & \multirow{2}{*}{Seq.\ skip list ETT}  & link & .000572  & .00520  & .107  & 1.06   & 9.17  & 45.5 \\
                                                           &                                       & cut  & .00110   & .0106  & .105  & 1.05   & 9.06   & 45.9 \\
   \cline{2-9}                                             & \multirow{2}{*}{Splay ETT}            & link & .000326  & .00311  & .0461  & .594   & 6.14   & 35.3 \\
                                                           &                                       & cut  & .000942  & .00890  & .0856  & .839   & 7.20   & 39.4 \\
   \cline{2-9}                                             & \multirow{2}{*}{ST-tree}              & link & .0000780 & .000975  & .0115  & .171   & 1.79   & 9.17 \\
                                                           &                                       & cut  & .000298   & .00284   & .0263  & .259   & 2.09   & 9.62 \\
  \hline
  \multirow{10}{*}{\makecell[cl]{Star\\graph}}             & \multirow{2}{*}{Parallel ETT (72(h))} & link & .000190  & .000465  & .00198 & .00558  & .0408  & .359 \\
                                                           &                                       & cut  & .000221  & .000475  & .000922 & .00323  & .0271  & .251 \\
   \cline{2-9}                                             & \multirow{2}{*}{Parallel ETT (1)}     & link & .000182  & .00170  & .0154  & .357   & 3.49   & 32.7 \\
                                                           &                                       & cut  & .000196  & .00194  & .0170  & .289   & 2.60   & 23.1 \\
   \cline{2-9}                                             & \multirow{2}{*}{Seq.\ skip list ETT}  & link & .000189  & .00212  & .0197  & .293   & 2.97   & 29.4 \\
                                                           &                                       & cut  & .000323  & .00317  & .0303  & .311   & 3.15   & 30.1 \\
   \cline{2-9}                                             & \multirow{2}{*}{Splay ETT}            & link & .000151  & .00124  & .0131  & .207   & 2.09   & 19.7 \\
                                                           &                                       & cut  & 2.85     & 2.85   & 2.86  & 3.07   & 5.09   & 25.9 \\
   \cline{2-9}                                             & \multirow{2}{*}{ST-tree}              & link & .0000150 & .000144 & .00123 & .0169  & .257   & 2.56 \\
                                                           &                                       & cut  & .0000379 & .000367 & .00336 & .0325  & .323   & 3.24 \\
  \bottomrule
\end{tabular}
  \caption[Dynamic trees data structure running times with varying batch size]{Running time (in seconds) of
  dynamic trees data structures on various graphs with $n = 10^7$.}
  \label{tab:ett-batch}
\end{table}
\begin{figure} \centering
  \begin{subfigure}{0.45 \textwidth} \centering
    \begin{tikzpicture}[scale=0.7]
    \begin{axis}[
    legend style={font=\tiny},
    xmode=log,
    ymode=log,
    xlabel={\small Batch size},
    ylabel={\small Running time (seconds)},
    xtick={1e2,1e3,1e4,1e5,1e6,1e7},
    xminorticks=false,
    legend pos=north west,
    ]
      \addplot[mark=*,color=red] coordinates {
        (1e2,   .000276)
        (1e3,   .000642)
        (1e4,   .0029)
        (1e5,   .0165)
        (1e6,   .0813)
        (1e7-1,   .305)
      };
      \addlegendentry{Parallel ETT (72(h))}

      \addplot[mark=x,color=blue,loosely dashed] coordinates {
        (1e2,   .00095)
        (1e3,   .00784)
        (1e4,   .175)
        (1e5,   1.29)
        (1e6,   7.25)
        (1e7-1,   29.2)
      };
      \addlegendentry{Parallel ETT (1)}

      \addplot[mark=o,color=olive, thick, densely dotted] coordinates {
        (1e2,   .000767)
        (1e3,   .00617)
        (1e4,   .131)
        (1e5,   1.04)
        (1e6,   6.77)
        (1e7-1,  33.4)
      };
      \addlegendentry{Seq.\ skip list ETT}

      \addplot[mark=square,color=purple,dashdotted] coordinates {
        (1e2,   .000408)
        (1e3,   .00345)
        (1e4,   .0631)
        (1e5,   .605)
        (1e6,   4.82)
        (1e7-1,  27.3)
      };
      \addlegendentry{Splay ETT}

      \addplot[mark=triangle,color=brown,densely dashed] coordinates {
        (1e2,   .000042)
        (1e3,   .000418)
        (1e4,   .00499)
        (1e5,   .0818)
        (1e6,   1.04)
        (1e7-1,  5.09)
      };
      \addlegendentry{ST-tree}

      \addplot[mark=none,color=lightgray,solid] coordinates {
        (1e2,   .576)
        (1e3,   .576)
        (1e4,   .576)
        (1e5,   .576)
        (1e6,   .576)
        (1e7-1, .576)
      };
      \addlegendentry{Static conn.\ (72(h))}
    \end{axis}
  \end{tikzpicture}
  \caption{Path graph, link}
  \end{subfigure}
  \begin{subfigure}{0.45 \textwidth} \centering
    \begin{tikzpicture}[scale=0.7]
    \begin{axis}[
    legend style={font=\tiny},
    xmode=log,
    ymode=log,
    xlabel={\small Batch size},
    ylabel={\small Running time (seconds)},
    xtick={1e2,1e3,1e4,1e5,1e6,1e7},
    xminorticks=false,
    legend pos=north west,
    ]
      \addplot[mark=*,color=red] coordinates {
        (1e2,   .000287)
        (1e3,   .000714)
        (1e4,   .00267)
        (1e5,   .0169)
        (1e6,   .098)
        (1e7-1,   .408)
      };
      \addlegendentry{Parallel ETT (72(h))}

      \addplot[mark=x,color=blue,loosely dashed] coordinates {
        (1e2,   .00223)
        (1e3,   .0187)
        (1e4,   .185)
        (1e5,   1.4)
        (1e6,   8.74)
        (1e7-1,   37.6)
      };
      \addlegendentry{Parallel ETT (1)}

      \addplot[mark=o,color=olive, thick, densely dotted] coordinates {
        (1e2,   .00148)
        (1e3,   .0144)
        (1e4,   .129)
        (1e5,   1.03)
        (1e6,   6.73)
        (1e7-1,  35.0)
      };
      \addlegendentry{Seq.\ skip list ETT}

      \addplot[mark=square,color=purple,dashdotted] coordinates {
        (1e2,   .00109)
        (1e3,   .0103)
        (1e4,   .0922)
        (1e5,   .776)
        (1e6,   5.51)
        (1e7-1,  31.7)
      };
      \addlegendentry{Splay ETT}

      \addplot[mark=triangle,color=brown,densely dashed] coordinates {
        (1e2,   .000562)
        (1e3,   .00471)
        (1e4,   .0405)
        (1e5,   .31)
        (1e6,   1.88)
        (1e7-1,  7.27)
      };
      \addlegendentry{ST-tree}
    \end{axis}
  \end{tikzpicture}
  \caption{Path graph, cut}
  \end{subfigure}
  \begin{subfigure}{0.45 \textwidth} \centering
    \begin{tikzpicture}[scale=0.7]
    \begin{axis}[
    legend style={font=\tiny},
    xmode=log,
    ymode=log,
    xlabel={\small Batch size},
    ylabel={\small Running time (seconds)},
    xtick={1e2,1e3,1e4,1e5,1e6,1e7},
    xminorticks=false,
    legend pos=north west,
    ]
      \addplot[mark=*,color=red] coordinates {
        (1e2,   .000208)
        (1e3,   .000584)
        (1e4,   .00279)
        (1e5,   .0153)
        (1e6,   .0972)
        (1e7-1,   .43)
      };
      \addlegendentry{Parallel ETT (72(h))}

      \addplot[mark=x,color=blue,loosely dashed] coordinates {
        (1e2,   .000674)
        (1e3,   .0056)
        (1e4,   .144)
        (1e5,   1.25)
        (1e6,   8.77)
        (1e7-1,   42.3)
      };
      \addlegendentry{Parallel ETT (1)}

      \addplot[mark=o,color=olive, thick, densely dotted] coordinates {
        (1e2,   .000572)
        (1e3,   .0052)
        (1e4,   .107)
        (1e5,   1.06)
        (1e6,   9.17)
        (1e7-1,  45.5)
      };
      \addlegendentry{Seq.\ skip list ETT}

      \addplot[mark=square,color=purple,dashdotted] coordinates {
        (1e2,   .000326)
        (1e3,   .00311)
        (1e4,   .0461)
        (1e5,   .594)
        (1e6,   6.14)
        (1e7-1,  35.3)
      };
      \addlegendentry{Splay ETT}

      \addplot[mark=triangle,color=brown,densely dashed] coordinates {
        (1e2,   .000078)
        (1e3,   .000975)
        (1e4,   .0115)
        (1e5,   .171)
        (1e6,   1.79)
        (1e7-1,  9.17)
      };
      \addlegendentry{ST-tree}

      \addplot[mark=none,color=lightgray,solid] coordinates {
        (1e2,   .174)
        (1e3,   .174)
        (1e4,   .174)
        (1e5,   .174)
        (1e6,   .174)
        (1e7-1, .174)
      };
      \addlegendentry{Static conn.\ (72(h))}
    \end{axis}
  \end{tikzpicture}
  \caption{Random recursive tree, link}
  \end{subfigure}
  \begin{subfigure}{0.45 \textwidth} \centering
    \begin{tikzpicture}[scale=0.7]
    \begin{axis}[
    legend style={font=\tiny},
    xmode=log,
    ymode=log,
    xlabel={\small Batch size},
    ylabel={\small Running time (seconds)},
    xtick={1e2,1e3,1e4,1e5,1e6,1e7},
    xminorticks=false,
    legend pos=north west,
    ]
      \addplot[mark=*,color=red] coordinates {
        (1e2,   .000237)
        (1e3,   .000599)
        (1e4,   .00195)
        (1e5,   .013)
        (1e6,   .0882)
        (1e7-1,   .42)
      };
      \addlegendentry{Parallel ETT (72(h))}

      \addplot[mark=x,color=blue,loosely dashed] coordinates {
        (1e2,   .00113)
        (1e3,   .01)
        (1e4,   .137)
        (1e5,   1.12)
        (1e6,   7.82)
        (1e7-1,   37.9)
      };
      \addlegendentry{Parallel ETT (1)}

      \addplot[mark=o,color=olive, thick, densely dotted] coordinates {
        (1e2,   .00103)
        (1e3,   .00880)
        (1e4,   .0981)
        (1e5,   1.08)
        (1e6,   9.79)
        (1e7-1,  56.4)
      };
      \addlegendentry{Seq.\ skip list ETT}

      \addplot[mark=square,color=purple,dashdotted] coordinates {
        (1e2,   .0011)
        (1e3,   .0106)
        (1e4,   .105)
        (1e5,   1.05)
        (1e6,   9.06)
        (1e7-1,  45.9)
      };
      \addlegendentry{Splay ETT}

      \addplot[mark=triangle,color=brown,densely dashed] coordinates {
        (1e2,   .000298)
        (1e3,   .00284)
        (1e4,   .0263)
        (1e5,   .259)
        (1e6,   2.09)
        (1e7-1,  9.62)
      };
      \addlegendentry{ST-tree}
    \end{axis}
  \end{tikzpicture}
  \caption{Random recursive tree, cut}
  \end{subfigure}
  \begin{subfigure}{0.45 \textwidth} \centering
    \begin{tikzpicture}[scale=0.7]
    \begin{axis}[
    legend style={font=\tiny},
    xmode=log,
    ymode=log,
    xlabel={\small Batch size},
    ylabel={\small Running time (seconds)},
    xtick={1e2,1e3,1e4,1e5,1e6,1e7},
    xminorticks=false,
    legend pos=north west,
    ]
      \addplot[mark=*,color=red] coordinates {
        (1e2,   .00019)
        (1e3,   .000465)
        (1e4,   .00198)
        (1e5,   .00558)
        (1e6,   .0408)
        (1e7-1,   .359)
      };
      \addlegendentry{Parallel ETT (72(h))}

      \addplot[mark=x,color=blue,loosely dashed] coordinates {
        (1e2,   .000182)
        (1e3,   .0017)
        (1e4,   .0154)
        (1e5,   0.357)
        (1e6,   3.49)
        (1e7-1,   32.7)
      };
      \addlegendentry{Parallel ETT (1)}

      \addplot[mark=o,color=olive, thick, densely dotted] coordinates {
        (1e2,   .000189)
        (1e3,   .00212)
        (1e4,   .0197)
        (1e5,   0.293)
        (1e6,   2.97)
        (1e7-1,  29.4)
      };
      \addlegendentry{Seq.\ skip list ETT}

      \addplot[mark=square,color=purple,dashdotted] coordinates {
        (1e2,   .000151)
        (1e3,   .00124)
        (1e4,   .0131)
        (1e5,   .207)
        (1e6,   2.09)
        (1e7-1,  19.7)
      };
      \addlegendentry{Splay ETT}

      \addplot[mark=triangle,color=brown,densely dashed] coordinates {
        (1e2,   .000015)
        (1e3,   .000144)
        (1e4,   .00123)
        (1e5,   .0169)
        (1e6,   0.257)
        (1e7-1,  2.56)
      };
      \addlegendentry{ST-tree}

      \addplot[mark=none,color=lightgray,solid] coordinates {
        (1e2,   .113)
        (1e3,   .113)
        (1e4,   .113)
        (1e5,   .113)
        (1e6,   .113)
        (1e7-1, .113)
      };
      \addlegendentry{Static conn.\ (72(h))}
    \end{axis}
  \end{tikzpicture}
  \caption{Star graph, link}
  \end{subfigure}
  \begin{subfigure}{0.45 \textwidth} \centering
    \begin{tikzpicture}[scale=0.7]
    \begin{axis}[
    legend style={font=\tiny},
    xmode=log,
    ymode=log,
    xlabel={\small Batch size},
    ylabel={\small Running time (seconds)},
    xtick={1e2,1e3,1e4,1e5,1e6,1e7},
    xminorticks=false,
    legend pos=north west,
    ]
      \addplot[mark=*,color=red] coordinates {
        (1e2,   .000221)
        (1e3,   .000475)
        (1e4,   .000922)
        (1e5,   .00323)
        (1e6,   .0271)
        (1e7-1,   .251)
      };
      \addlegendentry{Parallel ETT (72(h))}

      \addplot[mark=x,color=blue,loosely dashed] coordinates {
        (1e2,   .000196)
        (1e3,   .00194)
        (1e4,   .017)
        (1e5,   0.289)
        (1e6,   2.6)
        (1e7-1,   23.1)
      };
      \addlegendentry{Parallel ETT (1)}

      \addplot[mark=o,color=olive, thick, densely dotted] coordinates {
        (1e2,   .000323)
        (1e3,   .00317)
        (1e4,   .0303)
        (1e5,   0.311)
        (1e6,   3.15)
        (1e7-1,  30.1)
      };
      \addlegendentry{Seq.\ skip list ETT}

      \addplot[mark=square,color=purple,dashdotted] coordinates {
        (1e2,   2.85)
        (1e3,   2.85)
        (1e4,   2.86)
        (1e5,   3.07)
        (1e6,   5.09)
        (1e7-1,  25.9)
      };
      \addlegendentry{Splay ETT}

      \addplot[mark=triangle,color=brown,densely dashed] coordinates {
        (1e2,   .0000379)
        (1e3,   .000367)
        (1e4,   .00336)
        (1e5,   .0325)
        (1e6,   0.323)
        (1e7-1,  3.24)
      };
      \addlegendentry{ST-tree}
    \end{axis}
  \end{tikzpicture}
  \caption{Star graph, cut}
  \end{subfigure}
  \caption[Dynamic trees data structure running times with varying batch
  size]{Running time of dynamic trees data structure operations on trees
  of size $n=10^7$ with varying batch size.}
  \label{fig:ett-batch}%
\end{figure}
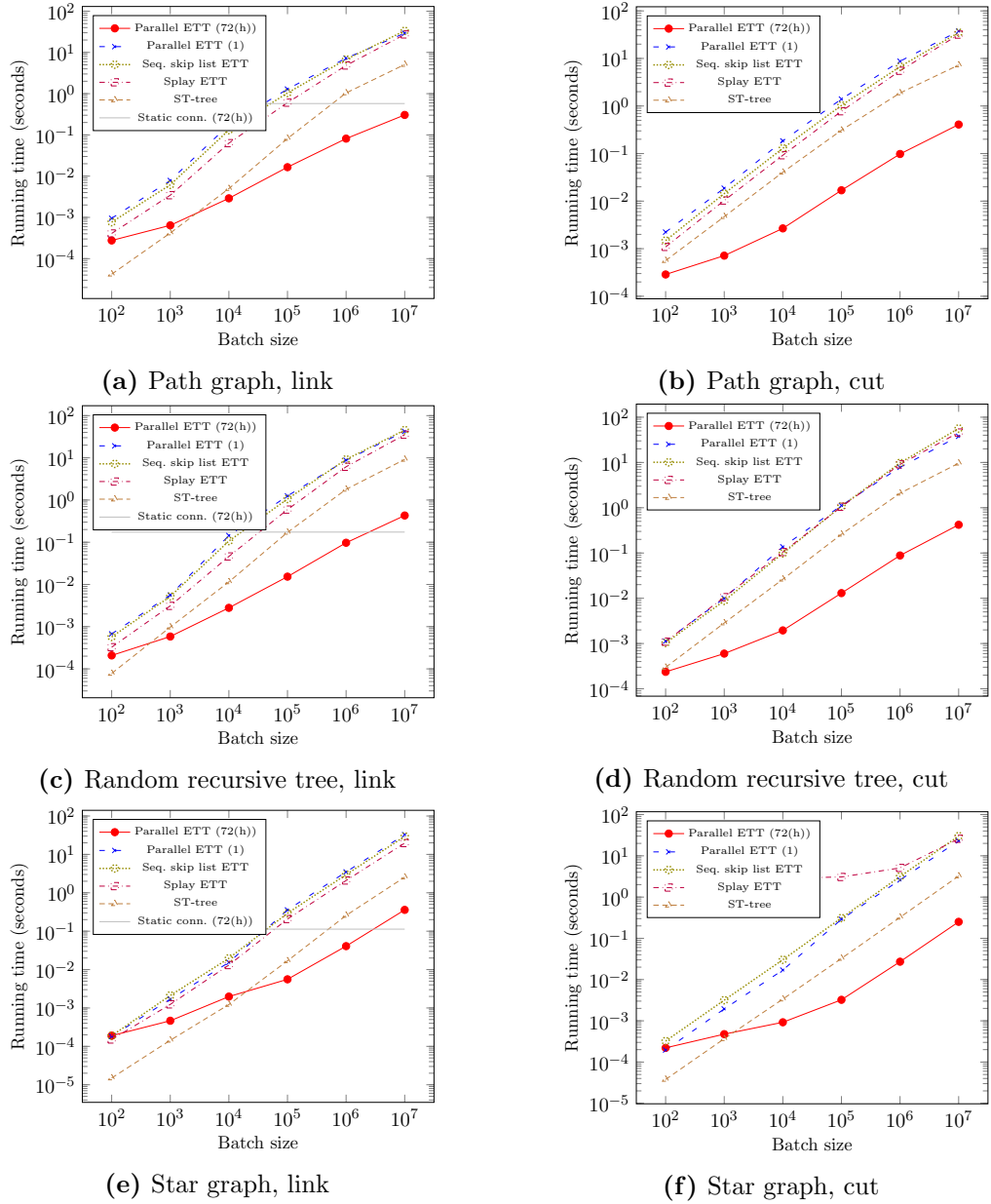%

Table~\ref{tab:ett-threads} and  Figure~\ref{fig:ett-threads} display the speedup of our parallel Euler tour
tree algorithms with a batch sizes of $k=10^4$ and $k=10^6$.
When running on $72$ cores with hyper-threading, we get good
speedup ranging from $82\times$ to $96\times$ for $k=10^6$ across all tested
graphs. For $k=10^4$, we found speedup ranging from $7.5\times$ to $75\times$
where the worst speedup was on the star graph. On the star graph, the
single-threaded running time is already fast for batch size
$k=10^4$, so there is not as much room for speedup.

In Table~\ref{tab:ett-batch} and Figure~\ref{fig:ett-batch}, we show the running times of our Euler tour tree
along with the times for sequential dynamic trees data structures. We also show
in Table~\ref{tab:static-conn} the time to run the static connectivity algorithm
on the full graphs for comparison. On large batch sizes, parallelism beats all
the sequential data structures, as expected. Though ST-trees are faster than
Euler tour trees sequentially and are unusually fast on the star graph due to
them performing well on graphs with small diameter, our parallel Euler tour tree
eventually outspeeds ST-trees on large batches even on the star graph.
(As an artifact of our testing setup, the splay-tree-based Euler tour tree
performs poorly on the star graph. The access pattern on the splay trees when
constructing the graph leads to a very deep splay tree, so the first few cut
operations after the graph construction setup are expensive.) In addition, the
performance of our Euler tour tree running on a single thread is similar to that
of conventional sequential Euler tour trees. We also see that the time to update
our Euler tour tree is much less than the time to statically compute
connectivity for all but the largest batch sizes.

\section{Discussion}

We showed that skip lists are a simple, fast data structure for parallel joining
and splitting of sequences and that we can use these skip lists to build a
batch-parallel Euler tour tree. Both of these data structures have good
theoretical efficiency bounds and achieve good performance in
practice. We hope that our work not only brings greater understanding
to the field of parallel data structures but also serves as a stepping
stone towards efficient parallel algorithms for dynamic graph
problems. In particular, Holm et al.\ give a dynamic connectivity
algorithm achieving $O\left(\log^2 n\right)$ amortized work per update
in which augmented Euler tour trees play a crucial
role~\cite{holm2001poly}. We believe that applying our parallel
Euler tour trees along with several other techniques can efficiently
parallelize Holm et al.'s algorithm. Obtaining a batch-parallel
solution for the general dynamic connectivity problem is an
interesting research problem that has not been adequately addressed in
the literature.

\section*{Acknowledgements}
Thanks to the reviewers for helpful comments.
This work was supported in part by NSF grants CCF-1408940, CCF-1533858, and CCF-1629444.

\bibliographystyle{abbrv}
\bibliography{ref}

\appendix

\section{Model}\label{sec:app-model}

The \emph{Multi-Threaded Random-Access Machine (\mpram{})} consists of a set of
\processes{} that share an unbounded memory.  Each \process{} is
essentially a Random Access Machine---it runs a program stored
in memory, has a constant number of registers, and uses standard RAM
instructions (including an \insend{} to finish the computation).  The
only difference between the \mpram{} and a RAM is the \forkins{} instruction that takes a positive
integer $k$ and forks $k$ new child \processes{}.  Each child \process{}
receives a unique identifier in the range $[1,\ldots,k]$ in its first
register and otherwise has the same state as the parent, which
has a $0$ in that register. All children start by running the next
instruction.  When a \process{} performs a \forkins{}, it is suspended until
all the children terminate (execute an \insend{} instruction).  A
computation starts with a single root \process{} and finishes when that
root \process{} terminates. This model supports what is often referred to as
nested parallelism. If the root \process{} never forks, it is a
standard sequential program.

We note that we can simulate the CRCW PRAM equipped with the same
operations with an additional $O(\log^{*} n)$ factor in the depth due
to load-balancing. Furthermore, a PRAM algorithm using $P$ processors
and $T$ time can be simulated in our model with $PT$ work and $T$
depth. We equip the model with a compare-and-swap operation (see
Section~\ref{sec:prelims}) in this paper.

Lastly, we define the cost-bounds for this model.  A computation can
be viewed as a series-parallel DAG in which each instruction is a
vertex, sequential instructions are composed in series, and the forked
sub\processes{} are composed in parallel. The \defn{work} of a
computation is the number of vertices and the \defn{depth}
(\defn{span}) is the length of the longest path in the DAG. We refer
the interested reader to~\cite{blelloch2018introduction-draft} for more details about
the model.

\section{Algorithm Implementation}\label{sec:app-implementations}

\subsection{Skip Lists}\label{sec:app-skip-list-implementation}

In our implementation of our skip lists, instead of representing an element of
height $h$ as $h$ distinct nodes, we instead allocate an array holding $h$
\textproc{left} and \textproc{right} pointers. This avoids jumping around in memory
when traversing up direct parents. In fact, we allocate an array whose size is
$h$ rounded up to the next power of two. This decreases the number of
distinct-sized arrays, which makes memory allocation easier when performing
concurrent allocation. We also cap the height at $32$, again for easier
allocation. We set the probability of a node having a direct parent to be $p = 1/2$.

We also need to be careful about read-write reordering on architectures with
relaxed memory consistency. For \textproc{Join} (Algorithm~\ref{alg:join}),
if the reads from the searches in lines~\ref{line:join-search-left}
to~\ref{line:join-search-right} are reordered before the write on
line~\ref{line:join-set-left}, then we can fail to find a parent link that should
be added. Thus we disallow reads from being reordered before
line~\ref{line:join-set-left}.

For augmented skip lists, instead of keeping $h$ \textproc{needs\_update}
booleans for each element, we keep a single integer that saves the
lowest level on which the element needs an update. This works because if a node
needs updating, then all its direct ancestors need updating as well.

Another optimization saves a constant factor in the work for batch split.
In \textproc{BatchUpdateValues}, we first walk up the skip list claiming nodes
and then walk back down to update all the augmented values. We perform these two
passes because in order to update a node's value, we need to know all of its
childrens' values are already updated too, which is easier to coordinate when
walking top-down through the list. However, in a batch split, after cutting up
the list, no nodes on the bottom level share any ancestors. As a result, we can
update all augmented values in a single pass walking up the list without even
using CAS.

\subsection{Euler Tour Trees}\label{sec:app-ett-implementation}

We implemented our parallel Euler tour tree algorithms, making several
adjustments for performance and for ease of implementation. For simplicity, we
use the unaugmented skip lists and do not support subtree queries.

To achieve good parallelism, we need to allocate and deallocate
skip list nodes in parallel. We use lock-free concurrent fixed-size allocators
that rely on both global and local pools. To reduce the number of fixed-size
allocators used, we constrain the skip list heights and arrays as described
in Appendix~\ref{sec:app-skip-list-implementation}.

For the dictionary $\codevar{edges}$, we use the deterministic hash table
dictionary from the Problem Based Benchmark Suite (PBBS)~\cite{shun2012brief}.
This hash table is based upon a phase-concurrent hash table developed by Shun and
Blelloch~\cite{shun2014phase}. As an additional storage optimization, for an
edge $\set{u, v}$ where $u < v$, we only store $(u,v)$ in our dictionary and use
the $\textproc{twin}$ pointer to look up $(v,u)$.

Instead of performing a semisort when batch joining, we found it faster to use
the parallel radix sort from PBBS.

For batch cut, we do not use list tail-finding because efficient list
tail-finding is challenging to implement. Instead, we opt for a
recursive batch cut algorithm.  Recall why we used list tail-finding
in Algorithm~\ref{alg:et-batch-cut}: we do not want to spend too much
time walking around adjacent edges to find one that is unmarked. In
our recursive batch cut algorithm, we resolve the issue by randomly
selecting a constant fraction of the edges from the input to ignore
and not cut. Then if we walk around adjacent edges naively, the number
of edges we need to walk around until we see an unmarked edge is
constant in expectation and $O(\log n)$ with high probability. Thus we
can cut out the unignored edges quickly. Then we recurse on the
ignored edges. This increases the depth by a factor of $O(\log k)$ but
does not asymptotically affect the work.

\section{Skip list correctness}\label{appendix:skip-list-correctness}

\subsection{Proof of
Theorem~\ref{thm:join-correctness} --- Concurrent Join Correctness}\label{appendix:join-correctness}

Here we want to prove that a phase of concurrent \textproc{Join}
(Algorithm~\ref{alg:join}) calls on unaugmented skip lists behaves
correctly.

We first set up several definitions. Start by fixing a finite set of
elements in an initial, valid skip list. All the \textproc{up} and
\textproc{down} pointers are fixed, but new \textproc{left} and
\textproc{right} pointers will be set throughout a phase of joins. As
such, we will identify the state of a list $L'$ with the set of
\textproc{left} and \textproc{right} links set in $L'$. For two nodes
$u$ and $v$, we write $(u, v)_R$ to refer to the link $u \rightarrow
\textproc{right} = v$ and $(u, v)_L$ to refer to the link $v
\rightarrow \textproc{left} = u$. We write $(u, v)$ to refer to $(u,
v)_L$ and $(u, v)_R$ collectively; for a set $S$, $(u,v) \in S$ if
$(u,v)_L \in S$ and $(u,v)_R \in S$.

We define an operator $\textsc{ToAdd}(K, L')$ that takes a set of links $K$ and a skip list
$L'$. If $K$ is a set of links that are to be added to $L$
through calls to \textproc{Join}, $\textsc{ToAdd}(K, L')$ defines what additional links at
higher levels \textproc{Join} should find and add to $L'$ as a result of adding
$K$. We define $\textsc{ToAdd}(K, L') = \bigcup_i \textsc{ToAdd}^{(i)}(K,L')$ level-by-level. For level $0$,
$\textsc{ToAdd}^{(0)}(K, L')$ consists of all level-$0$ links in $K$. For $i > 0$, a level-$i$
link $\ell = (v_L, v_R)$ is in $\textsc{ToAdd}^{(i)}(K, L')$ if
\begin{enumerate}[label=(\alph*)]
    \item $\ell \in K$, or
    \item $\ell \not\in L'$ and there is a rightward path from $v_L \rightarrow \textproc{down}$ to
      $v_R \rightarrow \textproc{down}$ consisting of links from $K$ and
      $\textsc{ToAdd}^{(i-1)}(K, L')$, and at least one link on this path is from
      $\textsc{ToAdd}^{(i-1)}(K, L')$.
\end{enumerate}
As an example, in the lower diagram in figure~\ref{fig:sl-join-split}, if the solid links are in
$L'$ and the dashed links at the bottom level are in $K$, then the dashed links
as a whole form $\textsc{ToAdd}(K, L')$.

Note that in the second condition of the definition of $\textsc{ToAdd}^{(i)}(K,L')$, we require
that at least one link on the path is from
from $\textsc{ToAdd}^{(i-1)}(K, L')$. This signifies that there is at least one pending link
$\ell'$ at level $i-1$ such that calling \textproc{Join} on $\ell'$ might find
$\ell$ when searching for the parent link of $\ell'$. This requirement also enforces that
$\textsc{ToAdd}(\emptyset, L') = \emptyset$.

We formally present the theorem we wish to prove:
\begin{theorem}
  Starting with a valid skip list $L$ and a set of initially disjoint join
  operations $J$ (given as a set of links to add at the bottom level of $L$),
  after all join operations in $J$ complete, the final skip list is $L_F = L
  \cup \textsc{ToAdd}(J, L)$.
\end{theorem}
We assume all instructions are linearizable and consider the linearized
sequential ordering of operations, stepping through them one by one.
Note that to simplify analysis, we assume we know all the client calls $J$
(calls not generated by a recursive call from line~\ref{line:join-recurse}) to
$\textproc{Join}$ that will occur in this phase. All these calls will be
considered to have begun but to have not yet executed
line~\ref{line:join-set-right}.

We prove this theorem by proving that an invariant always holds through a phase
of joins. We will need to
keep track of a set of active links $A_t$ that are being added at time $t$.
At time $t$, for an execution of
\textproc{Join}$(\codevar{v}_L, \codevar{v}_R)$ on link $\ell = (\codevar{v}_L,
\codevar{v}_R)$:
\begin{enumerate}[label=(\alph*)]
  \item\label{enum:at-begun} If the execution has not yet performed the CAS on
      line~\ref{line:join-set-right} and $v_L \rightarrow
      \textproc{right}$ is null (that is, the CAS could succeed),
      then $\ell$ is in $A_t$.
    \item\label{enum:at-left} If the execution has performed the CAS on line~\ref{line:join-set-right}
      but not the write on line~\ref{line:join-set-left}, then $\ell$ is in
      $A_t$.
    \item If the execution is within the searches on
      lines \ref{line:join-search-left}-\ref{line:join-search-right} and either
      search returns null, we reject $\ell$ from $A_t$. In fact, as soon the
      read happens on line~\ref{line:search-update-current} of the search
      that destines the search call to return null, we reject $\ell$ from $A_t$.
    \item\label{enum:at-search-finish} Otherwise, if the execution is within the searches on
      lines \ref{line:join-search-left}-\ref{line:join-search-right}, consider
      the nodes $p_L$ and $p_R$ that would be found by
      calling $\textproc{SearchLeft}(v_L)$ and $\textproc{SearchRight}(v_R)$
      relative to $L_F$. If $p_L$ and $p_R$ are both non-null and $p_L \rightarrow
      \textproc{right}$ is null (that is, the CAS at the next level could
      succeed), then $\ell$ is in $A_t$.
    \item As soon as the $\textproc{SearchRight}$ on
      \ref{line:join-search-right} returns, if variables $\codevar{parent}_L$
      and $\codevar{parent}_R$ are both non-null, immediately consider this
      execution as finished and consider $\textproc{Join}(\codevar{parent}_L,
      \codevar{parent}_R)$ as begun.
\end{enumerate}
Letting $L_t$ be the state of the linked list at time $t$ with $L_0 = L$, we
show that following invariant holds:
\[
  \text{for all $t$,} \quad L \cup \textsc{ToAdd}(J, L) = L_t \cup \textsc{ToAdd}(A_t, L_t).
\]
After all operations are done at some time step $T$, we will have $A_T =
\emptyset$ and hence $L \cup \textsc{ToAdd}(J, L) = L_T \cup \textsc{ToAdd}(A_T, L_T) =
L_T \cup \textsc{ToAdd}(\emptyset, L_T) = L_T$, showing that the final list $L_T$ is $L \cup \textsc{ToAdd}(J,
L) = L_F$ as desired. Since $L_F$ is fixed, we can proceed by showing that
the invariant holds initially and that $L_t \cup \textsc{ToAdd}(A_t, L_t)$ is constant as
time progresses.

We induct on $t$ to show the invariant holds at all times. The invariant is true
at time $0$ because $L_0 = L$ and $A_0 = J$. For $t > 0$, we consider each event
that can change $L_t$ and $A_t$ from $L_{t-1}$ and $A_{t-1}$. These events are
when we write to pointers on lines
\ref{line:join-set-right}-\ref{line:join-set-left} and when we finish searching
on lines \ref{line:join-search-left}-\ref{line:join-search-right}.

The first such event is a join succeeding on the CAS conditional statement on
line~\ref{line:join-set-right}, writing a pointer from $v_L$ to $v_R$ on some
level $i$.  At time $t-1$, $\ell = (v_L, v_R)$ was already in $A_{t-1}$ by
case~\ref{enum:at-begun} in the definition of $A_{t-1}$, and the link stays in
$A_t$ by case~\ref{enum:at-left}. We add a directional link $(v_L, v_R)_R$ to
$L_t$, which could remove other active links from $A_t$ who are in
case~\ref{enum:at-begun} or case~\ref{enum:at-search-finish}. Any other affected
active links who are in case~\ref{enum:at-begun} are equal to $(v_L, v_R)$, so
their removal does not affect $A_t$. Any affected active link $\ell'$ in
case~\ref{enum:at-search-finish} is a child link of $(v_L, v_R)$ who has already
set pointers to add $\ell' \in L_{t-1}$. Because $\ell'$ is already in
$L_{t-1}$, removing $\ell'$ from $A_t$ and thus from $\textsc{ToAdd}(A_t, L_t)$ does not
cause issues on level $i-1$. However, we should check that
$\textsc{ToAdd}(A_t, L_t)$ is unaffected at higher levels. Indeed, the
next link $\ell'$ could directly affect in the inductively-defined
$\textsc{ToAdd}(A_t, L_t)$ is its parent $\ell$, but $\ell$ stays in
$\textsc{ToAdd}^{(i)}(A_t, L_t)$ due to $\ell$ being added to $A_t$. Thus
$\textsc{ToAdd}^{(i)}(A_t, L_t)$ is unchanged, and inductively each
$\textsc{ToAdd}^{(k)}(A_t, L_t)$ for $k > i$ is unchanged as well. Therefore,
only $\ell'$ is removed from $\textsc{ToAdd}(A_t, L_t)$, which still preserves
the invariant. As a whole, $L_t$ is changed by adding $(v_L, v_R)_R$ (which is
acceptable because $(v_L, v_R) \in A_{t-1} \subseteq L_{t-1} \cup
\textsc{ToAdd}(A_{t-1}, L_{t-1})$ already), and $\textsc{ToAdd}(A_t, L_t)$ is
changed by possibly removing some children of $\ell'$ who were already in
$L_{t-1}$. With this, we see that the invariant is preserved with $L_t \cup
\textsc{ToAdd}(A_t, L_t) = L_{t-1} \cup \textsc{ToAdd}(A_{t-1}, L_{t-1})$.

Another such event is a join finishing the write to a pointer $v_R \rightarrow
\textproc{left} = v_L$ on line~\ref{line:join-set-left}. Let $\ell = (v_L,
v_R)$ be on level $i$. This event adds $\ell_L$ to $L_t$, and now $\ell \in L_t$ since
the opposite direction link $\ell_R$ was added at an earlier time when this join
completed the CAS on line~\ref{line:join-set-right}. We also know $\ell \in
A_{t-1}$ since this join was in case~\ref{enum:at-left} of the definition of
$A_{t-1}$ before time $t$. Now we enter case~\ref{enum:at-search-finish}.  There are a
few possibilities to consider here. The first is that one of $p_L$ or $p_R$ as
described in case~\ref{enum:at-search-finish} does not exist, in which case
we remove $\ell$ from $A_t$. The only effect this has on $\textsc{ToAdd}(A_t, L_t)$ is to
remove $\ell$ from $\textsc{ToAdd}(A_t, L_t)$; nothing in $\textsc{ToAdd}^{(i+1)}(A_t, L_t)$ depends on
$\ell$ since $\ell$ has no parent, so the only link removed is $\ell$ from
$\textsc{ToAdd}(A_t, L_t)$. Since $\ell \in L_t$, $L_t \cup \textsc{ToAdd}(A_t, L_t)$ is preserved. The
second possibility is that $p_L$ and $p_R$ exist, but $(p_L, p_R)_R \in L_t$
already. Again we remove $\ell$ from $A_t$, and again the invariant is
preserved so long as we show that $\textsc{ToAdd}^{(i+1)}(A_t, L_t)$ is unchanged and thus we
only remove $\ell$ from $\textsc{ToAdd}(A_t, L_t)$.  Removing $\ell$ can only affect
$\textsc{ToAdd}^{(i+1)}(A_t, L_t)$ by removing its parent $(p_L, p_R)$ from $
\textsc{ToAdd}^{(i+1)}(A_t, L_t)$. This can only happen if $(p_L, p_R)$ is in $\textsc{ToAdd}^{(i+1)}(A_{t-1},
L_{t-1})$ but $(p_L, p_R) \not\in A_{t-1}$. Then whichever execution of join set
$p_L \rightarrow \textproc{right}$ to non-null must have also completed the
corresponding write on line~\ref{line:join-set-left} because otherwise $(p_L,
p_R)$ would be in $A_{t-1}$. This means that $(p_L, p_R) \in L_{t-1}$, which
combined with $(p_L, p_R) \not\in A_{t-1}$ gives that $(p_L, p_R)$ could not
have been in $\textsc{ToAdd}^{(i+1)}(A_{t-1}, L_{t-1})$, a contradiction.  Hence
$\textsc{ToAdd}^{(i+1)}(A_t, L_t) = \textsc{ToAdd}^{(i+1)}(A_{t-1}, L_{t-1})$, so the invariant is
preserved. The last possibility is that $\ell$ stays in $A_t$, in which case the
invariant is still preserved.

A third event is that a join that wrote some link $\ell \in L_{t-1}$ finishes its searches on
line~\ref{line:join-search-right} and successfully passes the conditional
statement on line~\ref{line:join-if-search}.
This removes a link $\ell$ from $A_t$ and adds its parent link $\ell' = (v'_L,
v'_R)$ to $A_t$.
This removes $\ell$ from $\textsc{ToAdd}^{(i)}(A_t, L_t)$ but leaves $\textsc{ToAdd}^{(i+1)}(A_t, L_t)$
unchanged --- in order for the join's search to find $\ell'$, $\ell'$ must have
already been in $\textsc{ToAdd}^{(i+1)}(A_{t-1}, L_{t-1})$ by virtue of a rightward path from
$v'_L\rightarrow \textproc{down}$ to
$v'_R\rightarrow \textproc{down}$ passing through $\ell \in A_{t-1} \subset
\textsc{ToAdd}(A_{t-1}, L_{t-1})$. Then $\textsc{ToAdd}(A_t, L_t) =
\textsc{ToAdd}(A_{t-1}, L_{t-1}) \setminus \set{\ell}$,
and $\ell$ must already be in $L_{t-1} = L_t$. Thus the invariant is preserved.

The last interesting event that can occur is that a join that wrote some link $\ell$ on
level $i$ finishes its searches on
line~\ref{line:join-search-right} and fails the conditional statement on
line~\ref{line:join-if-search}. If this is because one of the parent nodes ($p_L$
and $p_R$) does not exist, then already $\ell \not\in A_t$ by
case~\ref{enum:at-search-finish} in the definition of $A_t$, and nothing is
changed. Otherwise this is because $\ell$ fails to find a parent node even those
both parent nodes exist in $L_F$. In this case we remove $\ell$ from $A_t$. We note
that $\ell \in L_t$, so once again, if we can show that
$\textsc{ToAdd}^{(i+1)}(A_{t}, L_{t}) = \textsc{ToAdd}^{(i+1)}(A_{t-1}, L_{t-1})$ and hence
$\textsc{ToAdd}(A_{t}, L_{t}) = \textsc{ToAdd}(A_{t-1}, L_{t-1}) \setminus\set{\ell}$, then the
invariant is preserved. The only way $\textsc{ToAdd}^{(i+1)}(A_{t}, L_{t})$ may change is if
$(p_L, p_R) \in \textsc{ToAdd}^{(i+1)}(A_{t-1}, L_{t-1})$ by virtue of there being a
rightward path from $p_L \rightarrow \textproc{down}$ to $p_R \rightarrow
\textproc{down}$, and removing $\ell$ from $A_t$ removes the last link from
$\textsc{ToAdd}^{(i)}(A_t, L_t)$ in this path. But this is impossible because if $\ell$ is
the last link from $\textsc{ToAdd}^{(i)}(A_t, L_t)$, then the remainder of the links on the
path are already in $L_t$, so the join should have been able to traverse this
path to find parents $p_L$ and $p_R$. Hence the invariant is preserved too in this final
case.

Note that the event that a join fails the conditional statement on the CAS on
line~\ref{line:join-set-right} changes neither $L_t$ nor $A_t$. This is because
we carefully define $A_t$ to not include such joins that are doomed to fail.

Finally, we must argue that this whole process terminates. Consider any
particular call to join. The whole set of nodes is bounded by some finite
height, and each recursive join call generated walks up one level, so the number
of recursive join calls is finite. It now suffices to show that each individual
call takes a finite number of steps. The only place the algorithm may loop
indefinitely is if the searches get stuck in a cycle, never succeeding on the
conditional on line~\ref{line:search-conditional} of the search algorithm
(Algorithm~\ref{alg:search-left}). In
particular, it never detects that $\codevar{current} = \codevar{v}$. Yet, if a
path out of $\codevar{v}$ leads to a cycle that does not contain
$\codevar{v}$ itself, then this contradicts that the links in $L_t \subseteq L_t
\cup \textsc{ToAdd}(A_t,L_t) = L_F$ are a subset
of those of the well-formed list $L_F$; a cycle at a level in a well-formed
list loops fully around the level.

This completes the proof that the invariant holds at all time steps, proving the
theorem.

\subsection{Proof of
Theorem~\ref{thm:split-correctness} --- Concurrent Split Correctness}\label{appendix:split-correctness}

Here we want to prove that a phase of concurrent \textproc{Split}
(Algorithm~\ref{alg:split}) calls on unaugmented skip lists behaves
correctly.

Fix a finite set of nodes in an
initial, valid skip list $L$. We will share much of the notation used in the
proof of correctness for join in Appenxdix~\ref{appendix:join-correctness}.

Each split operation is specified by a pointer to a node $v$, indicating the
intent to cut the list between $v$ and its successor. In our
analysis, we still want to talk about links, so we  write $\textsc{NextLink}(v)$ to indicate
the two directional link between $v$ and its successor in the original skip list $L$, with
$\textsc{NextLink}(v)_L$ and $\textsc{NextLink}(v)_R$ distinguishing between the directional links. For a set of
nodes $V$, we let $\textsc{NextLink}(V) = \bigcup_{v \in V} \textsc{NextLink}(v)$.

We let $P(v, L')$ denote $v$'s left parent relative to $L'$, which is the
result of $\textproc{SearchLeft}(v)$ on the list $L'$. For a set of nodes $V$,
we let $P(V, L') = \bigcup_{v \in V} \set{P(v, L')}$ where we let $\set{P(v,L')}
= \emptyset$ if $v$ does not have a left parent in $L'$.

Using the following rules, we define $\textsc{Reachable}(v, L')$ to represent the set of nodes,
which we will call \emph{reachable ancestors} of $v$, whose links
$\textproc{Split}(v)$ is capable of cutting in $L'$:
\begin{enumerate}[label=(\alph*)]
  \item If $v \rightarrow \textproc{right}$ is non-null in $L'$, then $v \in
    \textsc{Reachable}(v, L')$.
  \item\label{enum:split-r-inductive} If $u \in \textsc{Reachable}(v, L')$, $P(u, L')$ exists, and $P(u, L') \rightarrow
    \textproc{right}$ is non-null in $L'$, then $P(u, L') \in \textsc{Reachable}(v, L')$.
\end{enumerate}
For a set of nodes $V$, we let $\textsc{Reachable}(V, L') = \bigcup_{v \in V}
\textsc{Reachable}(v, L')$.

Now we formally present the theorem we wish to prove:
\begin{theorem}
  Starting with a well-formed skip list $L$ and a set of split operations
  operations $S$ (given as a set of nodes at the bottom level of $L$ whose links
  between their successor should be cut), after all split operations in
  $S$ complete, the final skip list is $L \setminus \textsc{NextLink}(\textsc{Reachable}(S, L))$.
\end{theorem}
Note that $\textsc{Reachable}(S, L)$ is the set of all ancestors of nodes of $S$, which is
precisely the set of nodes whose links to their successors we hope to remove in
a phase of splits.

We assume all instructions are linearizable and consider the linearized
sequential ordering of operations, stepping through them one by one.
Again, to simplify analysis, we assume that we know all the client calls $S$
(calls not generated by a recursive call from line~\ref{line:split-recurse}) to
$\textproc{Split}$ that will occur in this phase and that they all begin at time
$0$.

To prove this theorem, we define an invariant and show that it always holds.
This time, we will distinguish several sets of active nodes $B_t, C_t,
E_t$ that are being worked on at time $t$. The set $B_t$ represents split calls
that have not yet CASed, $C_t$ represents split calls that have CASed but have
not yet cleared the pointer in the opposite direction, and $E_t$ represents the
progress of a search for a parent. At time $t$, for an execution of
\textproc{Split}$(v)$:
\begin{enumerate}[label=(\alph*)]
    \item If the execution has not yet performed the CAS on
      line~\ref{line:split-set-right} and $\codevar{ngh}$ is non-null, then $v$
      is in $B_t$.
    \item If the execution finished the CAS on line~\ref{line:split-set-right}
      but has not yet cleared the pointer in the opposite direction
      on line~\ref{line:split-set-left}, then $v \in C_t$.
    \item If the execution finished the write on line~\ref{line:split-set-left},
      then we are in the search phase. Let $u$ be the latest value written to
      the local variable $\codevar{current}$ in the execution of \textproc{SearchLeft} (Algorithm~\ref{alg:search-left}). Then
      $u$ is in $E_t$. For convenience we consider $\codevar{current}$ to be
      written to as soon as the pointer read occurs in line~\ref{line:search-update-current}
      of \textproc{SearchLeft}, and consider $\codevar{current}$ to be null if
      the subsequent conditional on line~\ref{line:search-conditional} will
      fail.
    \item Once the recursive call begins on line~\ref{line:split-recurse}, this
      execution is finished.
\end{enumerate}
Letting $L_t$ be the state of the linked list at time $t$ with $L_0 = L$, we
show that following invariant holds:
\[
  \text{for all $t$,} \quad L \setminus \textsc{NextLink}(\textsc{Reachable}(S, L)) = L_t \setminus \textsc{ToDelete}_t
\]
where
\[
  \textsc{ToDelete}_t = \textsc{NextLink}(\textsc{Reachable}(B_t \cup P(C_t \cup E_t, L_t), L_t) \cup C_t).
\]
After all operations are done at some time step $T$, we have $B_T = C_T = E_T
= \emptyset$ and hence $L \setminus \textsc{NextLink}(\textsc{Reachable}(S, L)) = L_T \setminus \textsc{ToDelete}_T = L_T
\setminus \textsc{NextLink}(\textsc{Reachable}(\emptyset \cup P(\emptyset, L_T), L_T) \cup \emptyset) = L_T$, showing
that the final list is $L \setminus \textsc{NextLink}(\textsc{Reachable}(S, L))$
as desired. Since $L \setminus \textsc{NextLink}(\textsc{Reachable}(S,
L))$ is fixed, we proceed by showing that the invariant holds initially and that
$L_t \setminus \textsc{ToDelete}_t$ is constant.

We induct on $t$ to show the invariant holds at all times. At time $0$, we have
$L_0 = L$, $B_0 = S$, $C_0 = E_0 = \emptyset$, and $\textsc{ToDelete}_0 =
\textsc{NextLink}(\textsc{Reachable}(S \cup P(\emptyset, L), L) \cup \emptyset = \textsc{NextLink}(\textsc{Reachable}(S))$. Then $L_0 \setminus \textsc{ToDelete}_0
= L \setminus \textsc{NextLink}(\textsc{Reachable}(S))$ and the invariant holds.  For $t > 0$, we consider each
event that can change $L_t, B_t, C_t$, and $E_t$.

When a split on node $v$ succeeds on the CAS, we move a $v$ from $B_t$
to $C_t$, and we remove $\textsc{NextLink}(v)_R$ from $L_t$. Moving an element from $B_t$ to
$C_t$ when that element had a non-null successor at time $t-1$ does not affect $\textsc{ToDelete}_t$. Removing $\textsc{NextLink}(v)_R$ from $L_t$ has the effect of
possibly shrinking $\textsc{Reachable}(u, L_t)$ for various nodes $u$. If $\textsc{Reachable}(u, L_t)$ shrinks, it
is because removing $\textsc{NextLink}(v)_R$ from $L_t$ removes $v$ from $\textsc{Reachable}(u, L_t)$ and then
subsequently removes all nodes that come from repeatedly applying
rule~\ref{enum:split-r-inductive} from the definition of $\textsc{Reachable}$ to $v$. But these
removed nodes are exactly covered by $\textsc{NextLink}(\textsc{Reachable}(P(v, L_t), L_t) \cup \set{v}) \subseteq \textsc{ToDelete}_t$. In
other words, the loss of some of $u$'s reachable ancestors is acceptable because
$v$ can still reach those ancestors. We see then that we preserve $\textsc{ToDelete}_t$. Also,
since $\textsc{NextLink}(v)$ was already in $\textsc{ToDelete}_{t-1}$, $L_t \setminus \textsc{ToDelete}_t$ is preserved despite the
loss of $\textsc{NextLink}(v)_R$ from $L_t$.

When a split on node $v$ fails its CAS or fails the check that $\codevar{ngh}$ is
non-null, we remove a $v$ from $B_t$. However, this means that $v \rightarrow
\textproc{right}$ was null already before this time step, so $\textsc{Reachable}(v, L_{t-1}) =
\emptyset$. Hence removing $v$ from $B_t$ does not change $\textsc{ToDelete}_t$, and the invariant
is preserved.

When a split on node $v$ succeeds on clearing the pointer on
line~\ref{line:split-set-left}, we move $v$ from $C_t$ to $E_t$
and remove $\textsc{NextLink}(v)_L$ from $L_t$. Moving $v$ from $C_t$ to $E_t$ removes $\textsc{NextLink}(v)$
from $\textsc{ToDelete}_t$. This is exactly compensated by removing $\textsc{NextLink}(v)$ from $L_t$, which
happens as a result of this split clearing $\textsc{NextLink}(v)_L$ on this time step and
having cleared $\textsc{NextLink}(v)_R$ on some previous time step. Removing $\textsc{NextLink}(v)_L$ from $L_t$
may also invalidate $P(u, L_t)$ for various nodes $u$ and subsequently affect
$\textsc{Reachable}(\cdot, L_t)$ for other nodes. However, if $P(u, L_t)$ becomes non-existent as
a result of the removal of $\textsc{NextLink}(v)_L$, then $P(v, L_t) = P(v, L_{t-1}) = P(u,
L_{t-1})$. As a corollary, if a node $u$ loses reachable ancestors in $\textsc{Reachable}(u,
L_t)$ due to removing $\textsc{NextLink}(v)_L$, those lost ancestors are covered
by $\textsc{Reachable}(P(v, L_t), L_t)$. Then we see that these changes in $P(\cdot, L_t)$ and $\textsc{Reachable}(\cdot,
L_t)$ do not affect $\textsc{ToDelete}_t$, so as a whole $L_t \setminus \textsc{ToDelete}_t$ is preserved.

When we update $\codevar{current}$ in the search for a parent to something that
does not fail the conditional on line~\ref{line:search-conditional} of
$\textproc{SearchLeft}$, we remove a $v$ from $E_t$ where $v \rightarrow
\textproc{up}$ is null and add its predecessor $v \rightarrow \textproc{left}$
to $E_t$. Since $v$ has no up parent, $P(v) = P(v \rightarrow
\textproc{left})$, and we see $\textsc{ToDelete}_t$ is preserved.

When a search for a parent returns null and causes a split call to
exit without making a recursive call, we remove some node $u$ from $E_t$.
However, the fact that the search returned a null indicates that $P(u, L_{t-1})$
already did not exist, so the removal of $u$ from $E_t$ does not change $\textsc{ToDelete}_t$.
Thus the invariant is preserved.

Lastly, when a search for a parent gives a non-null node $p$ and leads to a recursive
call $\textproc{Split}(p)$, we add $p$ to $B_t$ and remove its direct child $v$
from $E_t$. Note that $P(v, L_t) = p$, so $\textsc{ToDelete}_t$ is
unchanged and the invariant is preserved.

Finally, we must argue that this whole process terminates. The argument is the
same as for join. Each recursive split call walks up one level in our
finite-height set of nodes. Each individual call cannot loop in its search
since, as the only pointer modifications are setting pointers to null, the links
of $L_t$ are always a subset of the well-formed list $L$.

\section{Skip list efficiency}\label{appendix:skip-list-efficiency}

Recall that in our skip lists, each node independently has a direct parent with
probability $0 < p < 1$.

\subsection{Work}
We prove a work bound of $O(k \log (1 + n/k))$ in expectation over a set of $k$
splits over $n$ elements for unaugmented skip lists. The same strategy proves
the bound for joins and for operations on augmented skip lists.

For a set of $k$ splits, we first show that $O(k \log (1 + n/k))$ links need to be cut
in expectation. Over the first $h = \ceil{\log_{1/p}(1 + n/k)}$ levels, each
split needs to remove at most $h$ links (one at each level), so there are
$O(kh)$ pointers to remove over the first $h$ levels. For each level $k > h$,
the number of links to remove is bounded by the number of nodes on the level. The
probability that a particular node has height at least $\ell$ is $p^{\ell-1}$,
so the expected number of nodes reaching level $\ell$ is $np^{\ell-1}$.
Then the number of links summed across all levels $\ell > h$ is at most
\begin{align*}
  n \sum_{\ell=h + 1}^\infty p^{\ell-1} = np^h \frac{1}{1 - p} \le np^{\log_{1/p}(1 + n/k)}
  \frac{1}{1 - p} \\
  = \frac{n}{(1-p)(1 + n/k)} \le
  \frac{n}{(1-p)(n/k)} = O(k).
\end{align*}
Therefore, the expected number of links we need to cut in total is $O(kh) + O(k)
= O(k \log (1 + n/k))$.

For each link to remove, the amount of work to find that link from the
previous child link in a split is $O(1)$ in expectation. To search for a link at
level $i + 1$ that needs removal, we call \textproc{SearchLeft}, which
walks left from the previous place we removed a link on level $i$ until we see a
direct parent. The amount of work is proportional to the number of nodes we
touch when walking left. The probability a node has a direct parent is
$p$ independently, so the number of nodes we must touch until we see a direct
parent is distributed according to $\on{Geometric}(p)$ with expected value $1/(1
- p) = O(1)$. (If we quit early due to reaching the beginning of a list or due
to detecting a cycle, that only reduces the amount of work we do.) Thus this
traversal work to find parent links only affects the expected work by a constant
factor.

Moreover, no two split operations can both remove the same link because we remove
links with a CAS. Whoever CASes the link first successfully clears the link, and
whoever comes afterwards quits. The quitting execution only does $O(1)$ expected extra
work from the extra traversal to find the already claimed link. Thus there is no
significant duplicate work per split.

Therefore, the work overall for $k$ splits is $O(k \log (1 +  n/k))$.

\subsection{Depth}
For analyzing depth, we know that with high probability, every search
path (a path from the top level of a skip list to a particular node on the
bottom level, or the reverse) in an $n$-element skip list has length $O(\log
n)$. A proof is given in~\cite{demaine2004lectureskip}. The main critical paths
of our operations consist of traversing search paths and doing up to a constant
amount of extra work at each step, so we get a depth bound of $O(\log n)$ with
high probability for any of our operations.

\end{document}